\documentclass[%
 reprint,
 superscriptaddress,
 groupedaddress,longbibliography,
 nofootinbib,
 amsmath,amssymb,
 aps,
 prx,
floatfix,
]{revtex4-2}
\usepackage{booktabs}
\usepackage{almpaper}

\DeclareFontFamily{U}{FdSymbolD}{}
\DeclareFontShape{U}{FdSymbolD}{m}{n}{<-> s * FdSymbolD-Book}{}
\DeclareSymbolFont{fdarrows}{U}{FdSymbolD}{m}{n}
\DeclareMathSymbol{\nperp}{\mathrel}{fdarrows}{225}

\newacronym{hnn}{HNN}{Hopfield neural network}
\newacronym[longplural={associative memories}]{am}{AM}{associative memory}
\newacronym[longplural={quantum associative memories}]{qam}{QAM}{quantum associative memory}
\newacronym{dfs}{DFS}{decoherence-free subspace}
\newacronym{gus}{GUS}{geometrically uniform states}
\newacronym{gksl}{GKSL}{Gorini–Kossakowski–Sudarshan–Lindblad}

\newcommand{\disp}[1]{\mathcal{D}[#1]}
\newtheorem{proposition}{Proposition}

\graphicspath{{figs/}}

\glsdisablehyper

\newcommand{\bh}{\cB(\hilbert)}

\begin{document}

\title{Theoretical framework for quantum associative memories}

\author{Adrià Labay-Mora}
\email{alabay@ifisc.uib-csic.es}
\author{Eliana Fiorelli}
\author{Roberta Zambrini}
\author{Gian Luca Giorgi}
\email{gianluca@ifisc.uib-csic.es}%
\affiliation{%
 Institute for Cross-Disciplinary Physics and Complex Systems (IFISC) UIB-CSIC, Campus Universitat Illes Balears, 07122 Palma de Mallorca, Spain.
}%

\date{\today}
\begin{abstract}
   Associative memory refers to the ability to relate a memory with an input and targets the restoration of corrupted patterns. It has been intensively studied in classical physical systems, as in neural networks where an attractor dynamics settles on stable solutions. Several extensions to the quantum domain have been recently reported, displaying different features. In this work, we develop a comprehensive framework for a quantum associative memory based on open quantum system dynamics, which allows us to compare existing models, identify the theoretical prerequisites for performing associative memory tasks, and extend it in different forms. The map that achieves an exponential increase in the number of stored patterns with respect to classical systems is derived. We establish the crucial role of symmetries and dissipation in the operation of quantum associative memory. Our theoretical analysis demonstrates the feasibility of addressing both quantum and classical patterns, orthogonal and non-orthogonal memories, stationary and metastable operating regimes, and measurement-based outputs. Finally, this opens up new avenues for practical applications in quantum computing and machine learning, such as quantum error correction or quantum memories.
\end{abstract}

\maketitle

\section{Introduction}

A system that is able to dynamically retrieve a set of pre-stored information can be generically referred to as \gls{am}, a concept that has its roots in neurophysiology \cite{hebb2005organization} and has been developed in the context of artificial intelligence. In 1982 a system was designed to function as an associative memory, the \gls{hnn} \cite{hopfield1982am}. It consists of an all-to-all network of classical spins, modeling neurons in active (+1) or inactive (-1) states, evolving to minimize a certain energy function through repeated network updates. This drives the system to settle into one of many stable configurations, the one associated with a stored memory, or pattern. The \gls{hnn} is indeed characterized by an attractor dynamics that enables the retrieval of a given pattern from a corrupted initial state. This features \glspl{am} as content-addressable memories, to be distinguished from Random-Access Memory where data is accessed based on specific addresses instead of content \cite{jaeger1997microelectronic}. A distinction can also be set between \glspl{am} and another common application of neural networks, such as classification. While \gls{am} focuses on retrieving patterns from distorted or incomplete inputs, classification tasks involve assigning inputs to specific categories based on learned features  \cite{zhang2000neural}.

In the quest to enhance and extend the capabilities of \glspl{am}, quantum realizations of these systems have been proposed. Indeed, in a broader context, the success of neural networks in diverse applications --such as image and speech recognition, natural language processing, and autonomous systems-- is driving innovation beyond classical settings establishing the burgeoning field of quantum machine learning \cite{biamonte2017quantum,cerezo2022challenges}. Recently, several different approaches have modeled quantum versions of \glspl{am}. The first proposals were reported in the nineties during the advent of quantum computing, mostly dealing with circuit-based approaches, and do not necessarily replicate the specific dynamics and functions of classical \gls{am}. Many of these digital models consist of variations of the Grover search algorithm \cite{ventura2000qam,trugenberger2001probabilistic,diamantini2006pattern}, or quantum implementations of perceptrons \cite{cao2017quantum,miller2021quantum}. Digital approaches have been employed in pattern classification tasks, including particle tracks in high-energy physics \cite{quiroz2021hep}, and genetic sequences \cite{lloyd2018qhnn}. While inspired by the classical \gls{hnn}, such models cannot be regarded as proper \glspl{am}, as they lack the association property, as we will discuss later.

Besides digital models, a second and more recent approach, which we refer to as analog, explores the dynamics of (open) quantum systems for realizing quantum instances of \glspl{am}. Here, generalizations of \glspl{hnn} range from two-level quantum systems to qudits, in both closed \cite{inoue2011pattern, glaser2009nuclear,das2023quantum} and open quantum systems \cite{rotondo2018open,fiorelli2019accelerated,fiorelli2021potts, bodeker2023optimal}. Some analog approaches deal with the derivation of effective \gls{am} models that exploit a quantum substrate. Examples include multimode Dicke-models  \cite{fiorelli2020signatures,carollo2021dicke} and confocal cavities QED systems \cite{MarshEtAl_ccqed_21}. These models embed patterns via classical learning rules. Additional works focus on unleashing the storage of \textit{quantum} patterns by exploiting quantum walks \cite{petruccione2014walks,china2019walks} or single driven-dissipative resonators \cite{labay2022memory,labay2023squeezed}. 

Alongside proposing different models and implementations for \gls{qam}, a major motivation in this emerging field is to understand the potential advantages of using quantum mechanics in these systems, and how quantum effects can improve their performance. Some current literature contributions focus on quantifying {the storage capacity, which refers to} the amount of information (memories) that can be stored in a system of a given size. Amongst the analog models that have been proposed, many of them operate in a vanishing storage capacity regime and deal with classical patterns \cite{rotondo2018open, fiorelli2020signatures, fiorelli2021potts}. These works account for quantum effects inducing, e.g, new dynamical phases \cite{rotondo2018open, fiorelli2021potts}, or speed-ups in the retrieval of information \cite{fiorelli2019accelerated}. The limits of storage capacity have recently been the subject of research in a number of different models that are presented as \gls{qam} instances \cite{MarshEtAl_ccqed_21, bodeker2023optimal, labay2022memory, lewenstein2021storage}. Some of the proposals do not exhibit any improvement compared to the classical counterpart \cite{bodeker2023optimal}, while other instances seem to identify a potential advantage \cite{labay2022memory, lewenstein2021storage}. Overall, this diverse collection of proposals is defining the emergent field of \gls{qam}, but a general framework that can describe and include the distinct instances of \glspl{qam} is still lacking. Consequently, performing meaningful comparisons amongst different models or identifying the potential of non-classical approaches remains a challenge. 

The objective of this work is to develop a comprehensive theory framework for \gls{qam}, beyond existing model-specific results, by providing a unified foundation for understanding the working principles of this function. Assuming a general approach, our starting point is the set of necessary properties that a generic open quantum system must show to be regarded as an associative memory. This will allow us to bound the capacity of quantum states that can be stored by these kinds of systems and compare it with their classical equivalents, to frame both classical and quantum patterns, and to establish the presence of symmetries, through the definition of basins of attraction, as the enabling mechanism for \glspl{qam}. Once the operative conditions for a \gls{qam} are identified, not only one can design \gls{qam} channels to store patterns in stable states but exploit the metastable phase in open quantum systems to store transient patterns. 

The work is organized as follows: in \cref{sec:sa} we review some key points of classical and quantum \glspl{am}, that we employ in \cref{sec:framework} to provide a general definition; the result of the latter allows us to build a general quantum channel, which is compatible with an \gls{am}, and it is given in \cref{sec:construct_qam} for both orthogonal and non-orthogonal memories.  The bound to the amount of information that can be stored in such a \gls{qam} is addressed in \cref{sec:sc}. The full characterization of systems that provide \gls{qam} is addressed in \cref{sec:symmetries,sect:meta}, respectively discussing the role of symmetries and the possible extension to metastable patterns. Finally, some physical instances of \gls{am} models are revisited within the introduced framework in \cref{sec:examples}, followed by conclusive discussion in \cref{sec:disc}.

\begin{figure}
    \centering
    \includegraphics[width=\linewidth]{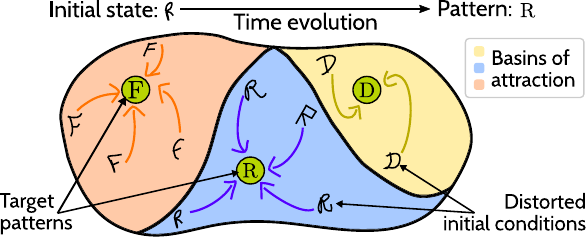}
    \caption{The phase space of a classical system powered with \gls{am} capabilities can be divided into different regions called basins of attraction. Each basin contains a stable fixed point of the dynamics, which is regarded as the pattern or memory. Then, any other initial condition belonging to the same basin will naturally converge to the corresponding pattern. Thus, the initial corrupted or distorted state is associated with the corrected version of the state, the pattern.}
    \label{fig:basins_classical}
\end{figure}

\section{Preliminaries} \label{sec:sa}

In this section, we introduce some preliminary concepts and report results on the storage capacity of classical and quantum \glspl{am}. In the classical regime, the \gls{hnn} \cite{hopfield1982am} is the most paradigmatic example of a general content-addressable system endowed with a number of fixed-point attractors, as anticipated in the Introduction. Patterns are locally stable states, each of them being related to a basin of attraction, a set of states that are dynamically evolved towards the corresponding memory. A sketch can be seen in \cref{fig:basins_classical}, where three fixed points encode letters that are used to correct distorted inputs. The collection of patterns is typically encoded via a learning rule in the neural connection, which, together with the dynamical evolution, characterizes the retrieval of information. These synaptic weights are chosen through a given dependency in terms of the states that serve as memories, one of the most employed being the Hebbian prescription \cite{amit1985spinglass}. Details about the \gls{hnn} dynamics and Hebbian rule can be found in \cref{apx:hopfield}.

An important figure of merit that characterizes \glspl{am} is the \textit{storage capacity}, defined as the maximum number of states that can be stored in a $n$-sized system,
\begin{equation} \label{eq:def_storage_capacity}
    \alpha = \frac{\text{ maximum number of patterns}}{\text{system size}}\ .
\end{equation}
For each \gls{am} with a given learning rule, the corresponding storage capacity $\alpha$ can be derived. For instance, a \gls{hnn}, equipped with uncorrelated patterns and Hebbian prescription, can store up to $0.138 n$ memories \cite{amit1985capacity}. When considering other learning rules, such as, e.g., the pseudo-inverse learning rule \cite{posner1987logcapacity} or nonlinear interactions \cite{demircigil2017model}, the corresponding limiting value $\alpha$ can increase, e.g., reaching $n^2/\ln{n^2}$ in the case of correlated patterns \cite{Willshaw1969}. The problem of storage capacity has been posed in more general terms in the seminal works of E. Gardner \cite{gardner1988capacity, gardner1988spaceinteraction}. Here, a set of patterns is imposed to be stationary states, while leaving the learning rule as a free parameter. As a result, a bound {is found} on the maximum value that the storage capacity can take for an entire class of \gls{hnn}-type \glspl{am}, i.e., irrespective of the specific learning prescription.  Such a bound  {is often}   referred to as a \textit{critical} storage capacity, and denoted as $\alpha_c$, to distinguish it from the storage capacity $\alpha$ calculated when a learning rule is defined. In Gardner's approach, the critical storage capacity of a \gls{hnn} is analyzed as a function of both the degree of pattern correlation and the size of the basin of attraction. In the limit of uncorrelated patterns and vanishing basin of attraction, the critical storage capacity reads $\alpha_c = 2 $, decreasing when enlarging the basin of attraction, and increasing when permitting correlated patterns \cite{gardner1988capacity}.

Concerning the issue of quantifying the storage capacity of \glspl{qam} {looking for a possible quantum advantage}, two main directions have been followed. Several research contributions tackle specific instances of \glspl{qam}. In this case, some results point out possible improvements with respect to the Hebbian limit \cite{MarshEtAl_ccqed_21, labay2022memory}, while other platforms are shown to behave similarly or worse than classical instances \cite{bodeker2023optimal}. Alternatively, recent contributions aim to provide more general bounds on the critical storage capacity of quantum \glspl{am} \cite{bodeker2023optimal, lewenstein2021storage}. In these scenarios, and in the same spirit of Gardner's program, one can define a quantum system undergoing retrieval dynamics while leaving the learning rule as a free parameter. In this respect, \ccite{lewenstein2021storage} shows that a quantum neural network behaving as a \gls{qam} can outperform the critical capacity of classical counterparts when renouncing any basin of attraction. 

As previously stated, establishing the extent of the applicability of the aforementioned outcomes remains a challenge, as a theoretical framework must still be defined. This is required to support both patterns as quantum states, similar to Refs. \cite{labay2023squeezed, lewenstein2021storage}, as well as memories exhibiting finite basins of attraction \cite{bodeker2023optimal, fiorelli2020signatures}. Tackling such an issue can allow us to shed light on $i)$ the general form of a \gls{qam}, and advance on the question as to whether the $ii)$ the critical storage capacity of the latter can outperform the classical \gls{am}.

In the two following sections, we combine the general approach that exploits the evolution of an open quantum system, as introduced in \ccite{lewenstein2021storage}, with Hopfield's original idea \cite{hopfield1982am} of dynamical systems displaying finite basin of attractions for each pattern. With these tools, we will tackle the issue $i)$ and $ii)$,  to characterize the properties and form of a \gls{cptp} channel for \gls{am}. Moreover, we will compute the storage capacity for different scenarios, particularly distinguishing the tasks where the retrieval of information is done with or without a measurement. 

\section{Theoretical framework} \label{sec:framework}

To address the limitations of current approaches and build a comprehensive {theory}, we start by framing the original definition of \gls{am} \cite{hopfield1982am} into the quantum formalism. Here, the pure (mixed) states of a physical system are represented by elements of a Hilbert space $\hilbert$ ($\bh$ space of bounded linear operators on $\hilbert$), and the dynamical evolution is described through a quantum channel, say $\Lambda$. This general formulation allows us to identify which properties and limitations characterize a generic (open) quantum system with $(\hilbert, \Lambda)$, that can be regarded as a \gls{qam}.

Quantum maps or channels are a key tool for describing the dynamics of quantum systems. They can be employed to formalize the continuous dynamics of open quantum systems undergoing Markovian evolution \cite{wolf2012quantum,breuer2002theory}, as well as discrete operations in quantum computation. Examples include noise effects, several types of qubit errors, and measurement processes \cite{nielsen2010quantum}. In general, a quantum channel $\Lambda$ is an operator that transforms a state $\rho \in \bh$ into another state $\Lambda(\rho) = \rho' $, where $ \rho' \in \cB(\hilbert')$ \footnote{In this work, we will focus on completely positive trace-preserving (\gls{cptp}) channels that map elements of $\bh$ to itself.}, the simplest example being a unitary evolution, $\Lambda(\rho) = U \rho U^\dagger$. 
The evolution of a state $\rho \in \bh$ by means of a \gls{cptp} map reads
\begin{equation}
    \Lambda(\rho) = \sum_{\alpha=1}^t K_\alpha \rho K_\alpha^\dagger\ ,
\end{equation}
where $t \le (\dim\hilbert)^2$, and  $\{ K_\alpha \}$ represents a set of Kraus operators, satisfying \cite{kraus1983states}
\begin{equation} \label{eq:general_cptp_cond}
    \sum_{\alpha=1}^t K_\alpha^\dagger K_\alpha = \Id\ .
\end{equation}
With the above definitions, we now identify the key properties and conditions that a quantum system with $(\hilbert,\Lambda)$ must possess to function as a \gls{qam}. 

As we anticipated, a classical \gls{am} requires a particular set of states to be stable fixed points of the dynamics. This guarantees that states representing the correct patterns are left unchanged by the dynamics and no information is lost (we will see a generalization in terms of metastable states in \cref{sect:meta}). We thus require that \textbf{(condition C1)} the set of $M$ states representing the patterns, $\{ \rho_\mu \in \bh \}_{\mu=1}^{M}$, are fixed points of the \gls{cptp} map  \cite{arias2002fppovm,wolf2012quantum,watrous2018theory}
\begin{equation} \label{eq:cptp_fixed_point}
    \Lambda(\rho_\mu) = \rho_\mu \ , \qquad \mu=1,\dots,M \ .
\end{equation}
For \gls{cptp} maps acting on finite-dimensional Hilbert spaces, condition \eqref{eq:cptp_fixed_point} admits at least one solution \cite{evans1977generators} where maps with just one fixed point are said to be ergodic \cite{burgarth2013ergodic}.  The existence of multiple fixed points has been extensively studied in the literature for both finite \cite{baumgartner2012structures} and infinite dimensional cases \cite{CarboneP_RepMP_2016}. Of course, any convex combination of these states is also a fixed point, an occurrence that we will further comment on at the end of the Section [see \cref{eq:def_spurious}]. To encompass the most general case of multiple fixed points, we may introduce the notion of maximal invariant subspace, $\mathcal{S}$, that, loosely speaking, represents the largest collection of states within the Hilbert space that remains unchanged under the action of the map. More precisely, given a state $\rho$ with support $\mathrm{supp}(\rho) \subset \mathcal{S}$, then $\mathrm{supp}[\Lambda(\rho)] \subset \mathcal{S}$, where  $\mathrm{supp}(X)$ is the set of eigenvectors of $X$ orthogonal to its kernel (i.e. with nonvanishing eigenvalues). As a note, eigenstates of unitary maps do not represent patterns as defined in \cref{eq:cptp_fixed_point}, except for the trivial identity map.

Secondly, we introduce the concept of decaying space $\cD = \cS^\perp$ complementary to the stable subspace $\cS$. Such subspace encompasses all quantum states $\ket\omega$ whose evolution under (many $r$) repeated actions of $\Lambda$ vanishes in $\cD$, i.e. $\cD = \{ \ket{\omega} \notin \mathcal{S} \mid  \lim_{r\to\infty} \Lambda^{r}(\ketbra{\omega})  \in \mathcal{B(\mathcal{S})}  \}$ \cite{baumgartner2012structures,CarboneP_RepMP_2016}. 
In other words, the states in the decaying subspace are mapped through $\Lambda$ into a state in the stable subspace. This is necessary but not sufficient for a \gls{qam} as we require that a sub-collection of states in $\cD$ are associated with a particular pattern $\rho_\mu$. To this end, we require that \textbf{(condition C2)} for each fixed point $\rho_\mu$ there exists a region of the Hilbert space, $\cD_{\mu}$ (decaying space), enclosing all the states that converge to the corresponding fixed point, $\rho_\mu$, under the action of the map $\Lambda$ (see \cref{fig:basins_quantum}). Then, analogous to the classical basin of attraction depicted in \cref{fig:basins_classical}, the quantum basin of attraction of a given fixed point is the combined set of the decaying space and the fixed point itself, i.e. $\cD_{\mu} \cup \{ \rho_\mu \}$. We stress the importance of modeling the presence of non-vanishing decaying spaces. Indeed, this feature excludes the trivial identity map from \gls{qam} and identifies the association property of the \gls{am}. We further note that beyond classical \gls{am}, its necessity is also recognized in (open) quantum systems performing quantum state classification tasks \cite{MarshallCVZ_19_PRA,schuld2014quest,petruccione2014walks}.

The presence of a decomposition of the Hilbert space into a stable subspace $\cS$, containing the patterns, and a decaying subspace $\cD$, containing the classification information, already restricts the  \gls{cptp} map. Indeed, the Kraus operators take the form  \cite{albert2019asymptotics}
\begin{equation} \label{eq:general_form_kraus}
    K_{\alpha} = \left( \begin{array}{c|ccc}
        K_{\alpha}^{\mathrm{S}} && K_{\alpha}^{\mathrm{SD}} & \\ \hline
         &&& \\
        0&& K_{\alpha}^{\mathrm{D}} & \\
        &&&
    \end{array} \right) .
\end{equation}
Here, $K_\alpha^{\mathrm{S}}$ ($K_\alpha^{\mathrm{D}}$) evolves the state in the invariant (decaying) subspace, and $K_\alpha^{\mathrm{SD}}$ plays the role of mapping the states from the decaying subspace into the stable one. We will refer to this term as cross-term. The zero block in the lower-left corner ensures that no information escapes from the stable {state space}.

\begin{figure}
    \centering
    \includegraphics[width=\linewidth]{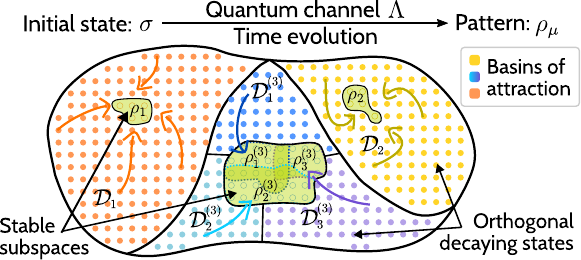}
    \caption{Quantum model for \gls{am}. The points represent the elements in a basis of the Hilbert space $\hilbert$, and the quantum states $\rho_\mu$ represent the patterns to be stored. Here, $\rho_1$ and $\rho_2$ are orthogonal patterns while the states $\rho_\ell^{(3)}$ represent non-orthogonal patterns. The figure highlights the division of the Hilbert space into decaying subspaces $\cD_\mu$, containing the initial conditions that converge towards the corresponding pattern.}
    \label{fig:basins_quantum}
\end{figure}

Finally, we enforce that \textbf{(condition C3)} any two states from different decaying spaces cannot be associated with the same memory $\rho_{\mu}$. To this end, we prevent the regions $ \lbrace \cD_{\mu} \rbrace$ from sharing common states, by imposing their mutual orthogonality \cite{ramsauer2020hopfield}. Although it may appear restrictive, we will see that this condition is fulfilled in physical quantum systems, like, e.g., those displaying symmetries \cite{buca2012symmetries,albert2014symmetries,minganti2018spectral}, as well as classical systems with multiple stable fixed points \cite{strogatz2018nonlinear}.

In summary, a quantum system with $(\hilbert, \Lambda)$ can function as a \gls{qam} if it meets the following requirements: (C1) patterns are quantum states $\rho_\mu$ that are fixed points of the map $\Lambda$ as in \cref{eq:cptp_fixed_point}; (C2) for each pattern $\rho_{\mu}$ there exists a non-vanishing decaying space, $\cD_\mu$; and (C3) the decaying spaces are mutually orthogonal. 
 
We can now formally define the necessary ingredients for a Quantum Associative Memory as follows:
\begin{definition} \label{def:qam}
 A system with Hilbert space $\hilbert$, and \gls{cptp} map $\Lambda$ functions as a \gls{qam} if there exists a set of $M$ fixed point of the map $\{ \rho_\mu \in \bh \}_{\mu=1}^M$, each one associated with a non-empty decaying subspace $\cD_\mu \subset \cS^\perp$, with $\cS$ the maximal invariant one. The subspaces $\cD_{\mu}$ are orthogonal one another, $\cD_\mu\perp\cD_\nu\ \forall\mu\neq\nu$, and
    \begin{equation} \label{eq:def_assoc_cond}
        \lim_{r\to \infty} \Lambda^r(\sigma) = \rho_\mu \qquad \forall \sigma\in \cB(\cD_\mu) \ .
    \end{equation}
\end{definition}

As commented before, each decaying subspace $\cD_\mu$ has a non-vanishing dimension, $\dim \cD_\mu =d_\mu > 0\ \forall\mu$, so each memory pattern can be reached from at least one decaying state. Moreover, since we imposed disjointness on the decaying spaces $\cD_\mu \cap \cD_\nu = \emptyset$ $\forall \mu \neq \nu$, the evolution does not allow states from different basins to mix~\cite{ramsauer2020hopfield}. It also follows, from the disjointness of the decaying spaces, that the total decaying subspace $\cD$ is the direct sum of the individual decaying subspaces $\cD_\mu$: $\cD = \bigoplus_{\mu=1}^M \cD_\mu$; where $M$ is the number of fixed points.

Although Definition \ref{def:qam} shares similarities with its classical counterpart, it is worth noticing that there is a fundamental difference between the two descriptions. While a classical \gls{hnn} evolves towards its local stable states through nonlinear dynamics, the dynamics of a quantum system is inherently linear due to the fundamental principles of quantum mechanics \cite{cerezo2022challenges}. As a consequence, given a \gls{qam} as defined above, any classical mixture of the set of patterns is also left unchanged by the action of the map. For instance, if we take two memory states $\rho_\mu$ and $\rho_{\nu}$, their classical mixture $ p \rho_\mu + (1 - p) \rho_\nu\ $ is an invariant state, as it is
\begin{equation} \label{eq:def_spurious}
    \Lambda[p \rho_\mu + (1 - p) \rho_\nu] = p \rho_\mu + (1 - p) \rho_\nu\ .
\end{equation}
Moreover, it is straightforward to see that each of the classical mixtures can be related to a new corresponding basin of attraction. Referring to the above example, any state belonging to $\cD_{\mu} \cup \cD_{\nu} $ is associated with a classical mixture of the respective patterns. As a result, for a \gls{qam} obeying the \cref{def:qam}, and displaying a set of $M$ patterns, any classical mixture of the latter behaves as an additional stable fixed point. By adopting the nomenclature employed in classical associative memories, classical mixtures of patterns can be referred to as \textit{spurious} memories -- states that are local minima but do not belong to the family of patterns. It is worth noticing that, even though classical \gls{am} models as the \gls{hnn} exhibit spurious patterns \cite{amit1985spinglass}, this occurrence is not systematic. Only certain combinations of classical patterns, and only in certain parametric regimes, will behave as stable fixed points of the dynamics. This is instead a common feature of all \glspl{qam} and we will discuss how this potential limitation can be handled in~\cref{sec:construct_qam}.

We also point out a distinctive capability of \gls{qam} that does not have any classical counterpart. Depending on the tasks that one delineates for the patterns upon retrieval, a \gls{qam} can support either a quantum output or a classical one. The first one refers to any set of (combination of) patterns that is retrieved by the \gls{qam} and does not require any further process, in this sense consisting of a quantum output. Suppose instead that one needs to access classical information. In that case, a measurement process can be further performed, leading to a probabilistic result of the patterns that matched the initial state. The measurement process may be seen as a nonlinear activation function which may produce an incorrect output~\cite{cao2017quantum}. However, the possibility of repeating the process a statistically significant number of times leads to additional information about the nature of the initial state, such as its similarity not only to the closest pattern, but also to the rest of the patterns \cite{schuld2014pattern}. We will comment on the differences between quantum and classical outputs in \cref{sec:sc}, where we will study the storage capacity properties of quantum associative memories. 

\section{Quantum maps for associative memory} \label{sec:construct_qam}

In this section, we provide an explicit construction of the quantum channel $\Lambda$ to realize a \gls{qam} according to \cref{def:qam} and discuss how a learning rule emerges in this formalism. For pedagogical reasons, we divide the analysis into two parts. First, in \cref{sec:orthogonal}, we build a map storing $M_\perp$ orthogonal quantum states $\{ \rho_\mu \}$ as fixed points; this is in analogy to classical models which require nearly orthogonal patterns for error-free retrieval \cite{hopfield1982am}. Then, in \cref{sec:nonrothogonal} we consider a more general problem where besides storing $M_\perp$ orthogonal patterns $\{ \rho_\mu \}$, we also allow for the storage of $M_\nperp$ non-orthogonal quantum patterns $\{ \rho_\ell \}$. Hence, we demonstrate how a quantum formulation of \gls{am} can store arbitrary quantum states as patterns, without restriction on the dimension of the stable subspace. The notation introduced in this section is summarized in \cref{tab:notation} for reference.

\begin{table}
    \centering
    \caption{Notation.}
    \label{tab:notation}
    \begin{tabular}{p{0.25\linewidth}p{0.7\linewidth}}
        \toprule 
        \textbf{Symbol} & \textbf{Definition} \\ \midrule
        $\hilbert$ & Complete Hilbert space \\
        $\cB(\bullet)$ & Space of bounded linear operators on $\bullet$ \\
        $\cS/\cD$ & Stable/Decaying subspace of $\hilbert$ \\
        $N^S/N^D$ & Dimension of stable/decaying subspace \\
        $N$ & Total Hilbert space dimension \\
        $\cS_\mu$ & Irreducible stable subspace \\
        $\cX_\tau$ & Decoherence-free subspace (DFS) \\
        $C_\perp / C_\nperp$ & Number of irreducible/decoherence-free subspaces \\
        $\rho_\mu$ & Orthogonal pattern spanning $\cS_\mu$ \\
        $\rho_\ell^{(\tau)}$ & Non-orthogonal pattern in $\cX_\tau$ \\
        $M_\perp / M_\nperp$ & Number of orthogonal/non-orthogonal patterns \\
        $m_\nperp^{(\tau)}$ & Number of non-orthogonal patterns in $\cX_\tau$ \\
        $M$ & Total number of patterns \\
        $\cD_\mu/\cD_\ell^{(\tau)}$ & Decaying subspace of pattern $\rho_\mu$/$\rho_\ell^{(\tau)}$ \\
        $s_\mu$/$s_\tau$ & Dimension of $\cS_\mu/\cX_\tau$ \\
        $d_\mu / d_\ell^{(\tau)}$ & Dimension of $\cD_\mu/\cD_\ell^{(\tau)}$ \\
        $\hilbert_\mu / \hilbert_\ell^{(\tau)}$ & Basin of attraction of $\rho_\mu$/$\rho_\ell^{(\tau)}$ \\
        $\ket{\mu_j} / \ket{\tau_j}$ & $j$-th basis element of $\cS_\mu/\cX_\tau$ \\
        $\ket*{\omega_x^{(\mu)}} / \ket*{\omega_x^{(\tau,\ell)}}$ & $x$-th basis element of $\cD_\mu/\cD_\ell^{(\tau)}$ \\ \bottomrule
    \end{tabular}
\end{table}

\subsection{Orthogonal patterns} \label{sec:orthogonal}

Let us consider the finite-dimensional Hilbert space $\hilbert$ with $ \dim\hilbert = N$, and let $\{ \rho_\mu \in \bh \}_{\mu=1}^M$ be the $M = M_\perp$ orthogonal patterns. The space features an orthogonal decomposition into stable and decaying parts such that $\hilbert = \cS \oplus \cD$ where, as commented in \cref{sec:framework}, the decaying part can be further decomposed into $M$ orthogonal blocks $\{ \cD_\mu \}_{\mu=1}^M$ (one for each pattern). Similarly, the stable subspace must be further decomposed into $M$ orthogonal blocks $\{\cS_\mu\}_{\mu=1}^M$ where each $\cS_\mu$ is the support of the pattern $\rho_\mu$, $\cS_\mu = \supp(\rho_\mu)$~\cite{CarboneP_RepMP_2016,baumgartner2012structures} to satisfy condition C1. These subspaces give rise to the maximal invariant subspace as $\cS = \bigoplus_\mu \cS_\mu \subset \hilbert$ where $ \dim \cS_\mu = s_\mu \ge 1$ and $\dim \cS  =  \sum_\mu s_\mu = N^S < N$ the corresponding dimensions. The complementary part of the Hilbert space which is not preserved by the map spans the decaying subspace of dimension $ \dim \cD = N - N^S = N^D$, where the dimension of each block is $d_\mu = \dim \cD_\mu$. Then, due to the orthogonality between blocks in both stable and decaying parts, we can define the basins of attraction for each pattern as $\hilbert_\mu = \cS_\mu \cup \cD_\mu$  such that $\hilbert = \bigoplus_{\mu=1}^M \hilbert_\mu$.

In our scenario, we want to guarantee that the association of a set of states to a specific pattern, say $\rho_\mu$, only occurs within subspaces corresponding to the same label $\mu$. As such, also the Kraus operators need to display a block structure, $K_\alpha = \bigoplus_\mu K_{\alpha,\mu}$, where each one of the block $K_{\alpha,\mu}$ has the form defined by \cref{eq:general_form_kraus}. Consequently, the map $\Lambda$, when restricted to the subspace $\hilbert_\mu$, has a unique fixed point $\rho_\mu$ \footnote{The map restricted to each basin of attraction $\cS_\mu$ is irreducible.}. Thus, the operators $K_{\alpha,\mu}$ must leave invariant the $\mu$-th subspace $\cS_\mu$, i.e., $K_{\alpha,\mu} \cS_\mu \subseteq \cS_\mu$~\cite{wolf2012quantum,burgarth2013ergodic}. This has been shown to imply the commutation relation $[K_{\alpha,\mu}, \rho_\mu] = 0\  \forall\alpha$~\cite{wolf2012quantum,watrous2018theory}. Moreover, the map further restricted to the stable subspace $\cS_\mu$ contains a unique, full-rank fixed point, $\rho_{\mu}$, and it can be shown that  $[K_{\alpha,\mu}^{\mathrm{S}},\rho_\mu] = 0$ \cite{albert2019asymptotics}. Hence, $\rho_\mu$ and the operators $K_{\alpha,\mu}^{\mathrm{S}}$ can be simultaneously diagonalized with respect to the same basis of eigenvectors, say $\{ \ket{\mu_j}\}_{j=1}^{s_\mu}$. In this basis, it is $\rho_\mu = \sum_{j=1}^{s_\mu} u_j^\mu \op{\mu_j}$, and we write for the operator on the stable subspace
\begin{equation} \label{eq:kraus_orthogonal_s}
    K_{\alpha,\mu}^{\mathrm{S}} = \sum_{j=1}^{s_\mu} a_{\mu,j}^\alpha \op{\mu_j}\ .
\end{equation}
This constraints each Kraus operator $K_{\alpha,\mu}^{\mathrm{S}}$ to the form of the pattern state. 

Let us now focus on the part of the Kraus operators $K_{\alpha,\mu}$ acting on the decaying subspace, which we called  $K_{\alpha,\mu}^{\mathrm{D}}$. Reminding that $\cD_{\mu} \perp \cS_{\mu}$,  and without loss of generality, we identify with $\{ \ket*{\omega_j^{\mu}} \}_j$, $j = 1,...,d_\mu,$ the orthonormal basis of $\cD_\mu$ in which  $K_{\alpha,\mu}^{\mathrm{D}}$ can be diagonalized,
\begin{equation}  \label{eq:kraus_orthogonal_d}    K_{\alpha,\mu}^{\mathrm{D}} = \sum_{x=1}^{d_\mu} c_{\mu,x}^\alpha \op{\omega_x^{\mu}} \ .
\end{equation}
It is worth noticing the label $\mu$ for states $\ket{\omega_j^{\mu}}$, which highlights that any state $\sigma$ belonging to the $\mu$-th decaying subspace, $\sigma \in \cB(\cD_\mu)$, is associated to the corresponding stable subspace $\cS_\mu$ through \cref{eq:def_assoc_cond}. Finally, the cross-term $K_{\alpha,\mu}^{\mathrm{SD}}$ must map any state from the decaying subspace $\cD_\mu$ to the stable one, $\cS_\mu$. That is, it needs to satisfy $K_{\alpha,\mu}^{\mathrm{SD}}\cD_\mu \subset \cS_\mu$. Therefore, referring to the basis $\{ \ket{\omega_x^{\mu}} \} $ and $\{ \ket{ \mu_j} \}$, the most general expression for the cross-term reads
\begin{equation}  \label{eq:kraus_orthogonal_sd}
    K_{\alpha,\mu}^{\mathrm{SD}} = \sum_{j=1}^{s_\mu} \sum_{x=1}^{d_\mu} b_{\mu, j,x}^\alpha \op{\mu_j}{\omega_x^{\mu}}\ .
\end{equation}

For \cref{eq:kraus_orthogonal_s,eq:kraus_orthogonal_sd,eq:kraus_orthogonal_d} to define a proper set of Kraus operators, we need to impose that any density operator, under the action of the map $\Lambda$, is evolved into a likewise valid density operator. Thus, the \gls{cptp} condition defined by \cref{eq:general_cptp_cond} must hold. Leaving the details of the derivation in \cref{apx:cptp_kraus}, we get
\begin{eqs}[eqs:kraus_orth_cptp_cond]
    &\sum_{\alpha} \abs{a_{\mu,j}^\alpha}^2 = 1\ , \label{eq:kraus_orth_cptp_cond_1} \\
    &\sum_\alpha (a_{\mu,j}^\alpha)^* b_{\mu,j,x}^\alpha = 0 \ , \label{eq:kraus_orth_cptp_cond_2} \\
    &\sum_\alpha \insqr{\sum_{j=1}^{s_\mu} (b_{\mu,j,y}^\alpha)^* b_{\mu,j,x}^\alpha} + \delta_{xy} \abs{c_{\mu,x}^\alpha}^2 = \delta_{xy}\ , \label{eq:kraus_orth_cptp_cond_3}
\end{eqs} 
which represent a series of constraints for the coefficients of~\cref{eq:kraus_orthogonal_s,eq:kraus_orthogonal_sd,eq:kraus_orthogonal_d}.

As a last step, it is worth noting that the form of the Kraus operators presented describes a system whose dynamics converge to a set of multiple steady states. However, the steady states are not uniquely defined as any state diagonal in the basis of $K_{\alpha,\mu}^{\mathrm{S}}$ is also a fixed point of the map. For this reason, we need to enforce the associativity condition \eqref{eq:def_assoc_cond} for a specific set of states $\{ \rho_\mu \}_{\mu=1}^{M}$, corresponding to the patterns. Hence, to guarantee that the states of the decaying subspace $\cD_\mu$ evolve to the corresponding pattern, $\rho_\mu$, the cross-term must satisfy
\begin{equation} \label{eq:kraus_assoc_cond}
    \sum_\alpha K_\alpha^{\mathrm{SD}} \op{\omega_x^\mu}{\omega_y^\mu} (K_\alpha^{\mathrm{SD}})^\dagger = \kappa_{xy}^\mu \rho_\mu \ ,
\end{equation}
where the constant $\kappa_{xy}^\mu$ determines the rate at which the operator $\op{\omega_x^\mu}{\omega_y^\mu}$ converges to the pattern $\rho_\mu$. 

The three conditions defined by~\cref{eq:kraus_orth_cptp_cond_1,eq:kraus_orth_cptp_cond_2,eq:kraus_orth_cptp_cond_3}, 
together with \cref{eq:kraus_assoc_cond}, determine the \textit{learning rule} which allows one to construct the \gls{cptp} map ruling a \gls{qam} with the quantum patterns $\{ \rho_\mu \}_{\mu=1}^M$. In \cref{apx:derivation_orthogonal} we continue the derivation of the Kraus operators by determining the value of the parameters in relation to the quantum patterns.

The \gls{cptp} map $\Lambda$ derived in this section enables an associative memory mechanism. This can be compared with \ccite{lewenstein2021storage} where a complete basis of the Hilbert space is assumed to be a set of fixed points of the map. However, imposing such a strict bound on the number of patterns implies no decaying space, i.e. the basins of attraction are reduced to the same memories. It is then concluded that the \gls{cptp} map realizing the association acts as a genuine incoherent operation \cite{marconi2022role}, i.e. as a decoherence map taking any quantum state to its classical mixture \cite{de2016gio}. It is possible to show that the map $\Lambda$ that we derived reduces to the one obtained in \cite{lewenstein2021storage}, upon restricting $\Lambda$ itself to the stable subspace (see \cref{apx:lewenstein} for details on the derivation). Moreover, our generalization permits us to endow each pattern with a finite-dimensional set of states that decay all and solely to the pattern itself. Any initial state displaying components in different basins of attraction will asymptotically evolve towards a convex superposition of patterns. However, the presence of finite-dimensional and disjoint decaying subspaces, $\cD_\mu$, permits the occurrence of pure association to one and only one pattern $\rho_\mu$.

\begin{example} \label{ex:amp_damp}
    \textit{Local amplitude damping}. To illustrate the concepts introduced so far, let us consider a simple example based on a modified amplitude damping channel acting on the subspace spanned by the pattern $\ket{\mu}$ and the decaying state $\ket{\omega_\mu}$. The Krauss operators of the map in such basis are $K_{0,\mu} = a_\mu^0\op{\mu} + \sqrt{1 - q_\mu} \op{\omega_\mu}$, $K_{1,\mu} = a_\mu^1 \op{\mu}$, and $K_{2,\mu} = \sqrt{q_\mu} \op{\mu}{\omega_\mu}$ where $q_\mu \in [0,1)$, $\sum_\alpha \abs{a_\mu^\alpha}^2 = 1$, and $a_\mu^\alpha \neq a_\nu^\alpha$ if $\mu\neq\nu$ (in such a way the Kraus operators are not proportional to the identity). Then, the complete Kraus operators are given by $K_\alpha = \bigoplus_\mu K_{\alpha,\mu}$. Here, a state $\sigma^{(0)} = \sum_{\mu,\nu} x_{\mu\nu} \op{\mu}{\nu} + y_{\mu\nu}\op{\omega_\mu}{\omega_\nu} + z_{\mu\nu}\op{\mu}{\omega_\nu} + z_{\mu\nu}^*\op{\omega_\mu}{\nu}$ transforms, under the action of the map, as
    \begin{align*}
        x_{\mu\nu} &\to \insqr{\sum_\alpha (a_\mu^\alpha)^* a_\nu^\alpha}x_{\mu\nu} + \sqrt{q_{\mu}q_\nu} y_{\mu\nu}\ , \\
        y_{\mu\nu} &\to y_{\mu\nu} \sqrt{(1-q_\mu)(1-q_\nu)}\ , \\
        z_{\mu\nu} &\to a_\mu^0 \sqrt{1 - q_\nu} z_{\mu\nu}\ .
    \end{align*}
    Thus, in the limit of infinite applications of the map, the state $\sigma^{(0)}$ is evolved toward the state $\sigma^{(\infty)} = \sum_\mu (x_{\mu\mu} + y_{\mu\mu})\op{\mu}$. This demonstrates the associative nature of the map, since the contribution of the decaying states $y_{\mu\mu}$ is driven towards the respective pattern. At the same time, the off-diagonal terms do not affect the outcome. Since they represent transitions between different basins of attraction (i.e., between different patterns $\ket{\mu}$ and $\ket{\nu}$), their presence is transient during the evolution. Still, they vanish in the limit of infinite applications of the map.
\end{example}

Beyond exhibiting the associative mechanism, the above example is also illustrative of the limitations exposed in \cref{sec:framework}. Indeed, the map leaves not only the intended memory patterns but also any convex combination of them (i.e., their mixtures) unchanged. As already commented, this occurrence can be considered the quantum equivalent of a spurious memory, a stable state that is not in the family of patterns. To overcome the presence of spurious memories in \gls{qam}, a potential strategy can be devised by leveraging quantum measurements as follows. In the case of orthogonal patterns, one can always define a projective measurement, $\{\hat{P}_\mu \}_{\mu=1}^M$, such that $\tr \hat{P}_\mu \rho_\nu = \delta_{\mu\nu}$. Then let us assume that an initial, corrupted state $\sigma^{(0)}$ whose probability to be found in a state corresponding to the pattern $\mu$ is $p_\mu^{(0)} = \tr \hat{P}_\mu \sigma^{(0)} = x_{\mu\mu}$. Via multiple applications of the map, the state $\sigma^{(0)}$  has been associated with a final, spurious state whose probability to be found in the pattern $\mu$ is $p_\mu^{(\infty)} = \tr \hat{P}_\mu \sigma^{(\infty)} = x_{\mu\mu} + y_{\mu\mu} > p_\mu^{(0)}$. Therefore, performing the quantum measurement will yield, with high probability, the pattern $\mu^*$ that best overlaps with the final state $\sigma^{(\infty)}$, up to a failure probability that depends on the contribution of the other patterns in the initial state.

\subsection{General formulation} \label{sec:nonrothogonal}

In the previous section, we studied the conditions that a quantum channel $\Lambda$ needs to satisfy, to permit the storage of states with orthogonal support.
In the following, we extend such analysis to the case of a \gls{qam} displaying general quantum states as patterns, thus including those with non-orthogonal support. Hence, we want to first identify the structure of the Hilbert space allowing for storage of $M_\perp$ orthogonal $\{ \rho_\mu \}$ ($\supp(\rho_\mu) \cap \supp(\rho_\nu) = \emptyset$ $\forall \mu, \nu$) and $M_\nperp$ non-orthogonal patterns $\{ \rho_\ell \}$ ($\supp(\rho_\ell) \cap \supp(\rho_{\ell'}) \neq \emptyset$ for some $\ell, \ell'$) to subsequently construct the \gls{cptp} map realizing the association.

\begin{figure}
    \centering
    \includegraphics[width=\linewidth]{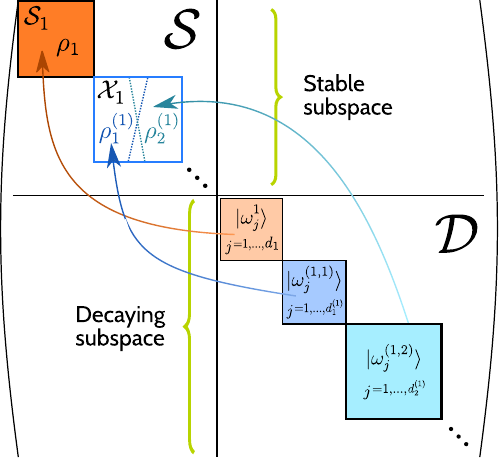}
    \caption{Action on quantum states of the most general form of the  QAM  map $\Lambda$ capable of storing orthogonal and non-orthogonal quantum patterns. Each orthogonal block in the stable subspace may correspond to an irreducible subspace $\cS_\mu = \supp \rho_\mu$ or a \acrlong{dfs} such that $\{ \rho_\ell^{(\tau)} \}_\ell \in \cB(\cX_\tau)$. The orthogonality of the states in the decaying subspace guarantees perfect association under the evolution of the map $\Lambda$ (represented by the arrows).}
    \label{fig:hilber_structure}
\end{figure}

Consistently with Definition \ref{def:qam}, the structure of the Hilbert space is again divided into two orthogonal blocks, $\cS$ and $\cD = \cS^\perp$, representing the invariant subspace, and the decaying one, respectively. With respect to the former, there appear $M_{\perp}$ fixed points of the map $\{\rho_{\mu}\}_{\mu=1}^{M_{\perp}}$, each of them spanning irreducible and orthogonal supports $\cS_\mu$. In addition, aiming at storing non-orthogonal patterns, we include a collection of $C_{\nperp}$ subspaces, $\cX_{\tau}$, each one preserving the coherences between their minimal invariant states, which hence exhibit non-orthogonal supports. In this general scenario, the decomposition of the invariant subspace reads \cite{CarboneP_RepMP_2016, baumgartner2012structures}
\begin{equation} \label{eq:structure_stable_subspace}
    \cS = \bigoplus_{\mu = 1}^{M_{\perp}} \cS_\mu \oplus \bigoplus_{\tau=1}^{ C_{\nperp}} \cX_{\tau} \ ,
\end{equation}
where both $\cS_\mu $, $\mu=1,...,M_{\perp}$ and $\cX_\tau$, $\tau = 1,..., C_{\nperp} $ are all mutually orthogonal subspaces ~\cite{baumgartner2012structures,amato2023asymptotics,CarboneP_RepMP_2016}. Each subspace $\cX_\tau$ can be in turn decomposed in terms of orthogonal subspaces, $\cX_\tau = \bigoplus_{j} \cW_{j, \tau} $. However, such a decomposition is not unique: there exists an isomorphism, $\cU_{\tau}$, mapping a given decomposition $\cW_{j, \tau}$ to an equivalent one, $ \cW'_{j, \tau}$, and thus relating all the minimal invariant states in $\cX_{\tau}$ one another. This in turn yields, in the most general case, subsystems in $\cX_\tau$ that are free from decoherence \cite{CarboneP_RepMP_2016, baumgartner2012structures}. In particular, any decomposition of $\cX_{\tau}$ in terms of $1$-dimensional $\cW_{j, \tau}$ corresponds to a \gls{dfs}, whereas any decomposition in terms of at least $2$-dimensional $\cW_{j, \tau}$ corresponds to a noiseless subsystem. In the following, we will focus on the \gls{dfs} case. It is worth also stressing that any other coherence, e.g., between $\cS_\mu$ and $\cX_\tau$, as well as between any two of either $\cS_\mu$ and $\cS_{\mu'}$, or $\cX_\tau$ and  $\cX_{\tau'}$ is not preserved by the evolution. For the sake of clarity, we write explicitly the full decomposition of the Hilbert space, $\hilbert= \bigoplus_{\mu=1}^{M_{\perp}} \cS_\mu \oplus \cD_\mu \bigoplus_{\tau=1}^{C_{\nperp}} \cX_\tau \oplus \cD_\tau$.

The presence of non-orthogonal patterns requires a more general definition of the \gls{qam} map $\Lambda$. The Kraus operators can be decomposed as $K_\alpha = \bigoplus_{\mu=1}^{M_{\perp}} K_{\alpha,\mu} \bigoplus_{\tau=1}^{C_\nperp} K_{\alpha,\tau}$. Here, the first part of the operators $\{ K_{\alpha,\mu} \}_{\mu =1}^{M_\perp}$ 
act on the related subspace $\bigoplus_{\mu =1}^{M_{\perp}} \cS_\mu \oplus  \cD_\mu$, so the results derived in \cref{sec:orthogonal} can be directly applied. The remaining part, $\{ K_{\alpha,\tau} \}_{\tau=1}^{C_{\nperp}}$, needs to be slightly reformulated to permit the storage of a given set of non-orthogonal patterns. Indeed, let us focus on the subspace $\cX_\tau$, with $\dim \cX_\tau = s_\tau$, and let $\{ \ket{{\tau}_j} \}_{j=1}^{s_\tau}$ be a basis of such a space. It is worth noticing that, as we deal with \gls{dfs}, each $\cX_\tau$ preserves any quantum state that can be written in a basis of the subspace. As a consequence, so long as there exists a non-empty decaying subspace $\cD_\tau$, and no other properties are enforced, any state belonging to $\cX_\tau$ behaves as a pattern. In the following, we aim instead at storing a \textit{specific} set of fixed points, say $\{ \rho_{\ell}^{(\tau)} \in \cB(\cX_\tau) \}_{\ell=1}^{m_{\nperp}^{(\tau)}}$. As such, notice that the total number of non-orthogonal patterns read $M_{\nperp} 
= \sum_{\tau=1}^{C_{\nperp}}m_{\nperp}^{(\tau)}$. Accordingly, a more detailed structure of the decaying subspace $\cD_\tau$ needs to be defined. 

In order to store the collection of fixed points $\rho_{\ell}^{(\tau)}$, $\ell= 1,...,m_{\nperp}^{(\tau)}$, which are in general neither orthogonal nor linearly independent, we associate, to each one of them, a decaying subspace $\cD_\ell^{(\tau)} = \sspan\{ \ket*{\omega_x^{(\tau, \ell)}} \}_{x=1}^{d_\ell^{(\tau)}}$. The states $ \{ \ket*{\omega_x^{(\tau,\ell)}} \}_{x=1}^{d_{\ell}}$ form an orthonormal basis, which is orthogonal to the stable subspace, consistently with~\cref{def:qam}. As a result, we can write the $\tau$-th decaying space through a block structure, 
\begin{equation}
\label{eq:structure_decaying_non-orthogonal}
\cD_\tau = \bigoplus_{\ell=1}^{m_{\nperp}^{(\tau)}}\cD_\ell^{(\tau)} \ .
\end{equation}
Whence, the most general expression for the Kraus operators acting on $\cX_\tau \oplus \cD_{\tau}$ is again given in terms of a block structure, as defined by \cref{eq:general_form_kraus}. The three main blocks,  $K_{\alpha, \tau}^{\mathrm{S}}$, $K_{\alpha, \tau}^{\mathrm{D}}$, which act on $\cX_\tau$, $\cD_{\tau}$, respectively, and the mixing one, $K_{\alpha, \tau}^{\mathrm{SD}}$, can be tackled separately.

Let us consider first the Kraus operator acting on the stable subspace. Here, because coherences between elements of the basis, $\op{{\tau}_j}{{\tau}_k}$, are preserved, and given that we impose the fixed point condition defined by \cref{eq:cptp_fixed_point}, the operator is proportional to the identity~\cite{albert2019asymptotics}
\begin{equation} \label{eq:kraus_gen_s}
K_{\alpha,\tau}^{\mathrm{S}} = a_{\tau}^\alpha \Id_{\tau} \ .
\end{equation}
The coefficients $a_{\tau}^\alpha$ need to satisfy the CPTP conditions \eqref{eq:kraus_orth_cptp_cond_1}, so it holds that $\sum_\alpha \abs{a_{\tau}^\alpha}^2 = 1$.

The Kraus operator acting on the decaying subspace $\cD_{\tau}$ can be further reduced. Indeed, given the decomposition in \cref{eq:structure_decaying_non-orthogonal}, it is $K_{\alpha,\tau}^{\mathrm{D}} = \bigoplus_{\ell=1}^{m_{\nperp}^{(\tau)}} K_{\alpha,\tau,\ell}^{\mathrm{D}}$, where
\begin{equation} \label{eq:kraus_gen_d}
    K_{\alpha,\tau,\ell}^{\mathrm{D}} = \sum_{x=1}^{d_\ell^{(\tau)}} c_{\ell,x}^{\alpha,\tau} \op{\omega_x^{(\tau, \ell)}}\ .
\end{equation}

Finally, we focus on the mixing term, $K_{\alpha, \tau}^{\mathrm{SD}}$, which plays a key role as it maps states belonging to the subspace $\cD_\ell^{(\tau)}$ into the corresponding pattern $\rho_{\ell}^{(\tau)}$. In general, we can write
\begin{equation} \label{eq:kraus_gen_sd}
    K_{\alpha,\tau}^{\mathrm{SD}} = \sum_{\ell=1}^{m_{\nperp}^{(\tau)}} \sum_{j=1}^{s_\tau} \sum_{x=1}^{d_\ell^{(\tau)}} b_{\ell,j,x}^{\alpha,\tau} \op{{\tau}_j}{\omega_x^{(\tau, \ell)}} \ ,
\end{equation}
which needs to satisfy the associativity condition
\begin{equation} \label{eq:kraus_gen_assoc}
    \sum_\alpha K_\alpha^{\mathrm{SD}} \op{\omega_x^{(\tau,\ell)}}{\omega_y^{(\tau, \ell')}} (K_\alpha^{\mathrm{SD}})^\dagger = \delta_{\ell\ell'} \kappa_{xy}^\ell \rho_{\ell}^{(\tau)}\ .
\end{equation}
We note that, in comparison to \cref{eq:kraus_assoc_cond} an extra delta function appears. This guarantees that cross-terms belonging to different decaying subspaces are suppressed, and do not evolve to stable states different from the patterns. In case the quantum patterns can be expressed as pure states, i.e. $\rho_\ell = \op{\psi_\ell}$, then it is possible to find an expression for the Kraus parameters in \cref{eq:kraus_gen_d,eq:kraus_gen_s,eq:kraus_gen_sd} as can be seen in \cref{apx:derivation_dfs}.

\section{Storage capacity} \label{sec:sc}

As anticipated in \cref{sec:sa}, an important quantity that characterizes \glspl{am} is the storage capacity. For a \gls{hnn}-type \gls{am} with $n$ neurons and $M$ patterns, it is given by the ratio $\alpha = M / n$ [see also \cref{eq:def_storage_capacity}]. This definition is based on the fact that classical patterns are stored as binary strings of length $n$, in which one can define at most $n$ orthogonal vectors. The expression for the storage capacity can be thus interpreted as the density of states that an \gls{am} can faithfully store (here $M$). The actual limit for a given model, i.e. the maximum storage capacity, depends on the given learning rule \cite{amit1985capacity,posner1987logcapacity,demircigil2017model}, and it is upper-bounded by the optimal or critical storage capacity, which is instead independent of the specific learning prescription \cite{gardner1988capacity, gardner1988spaceinteraction}, as summarized in \cref{sec:sa}. 

In the quantum realm, one would like to store quantum states, and subsequently quantify how many of them can be accommodated by a \gls{qam}. Here, in analogy with the number of classical orthogonal vectors that define the space of possible patterns, a meaningful quantity to consider is the Hilbert space dimension quantifying the number of possible orthogonal states of the system. For instance, a network of $n$ qubits features a $2^n$-dimensional Hilbert space. However, as we will see, not all of them can be simultaneously stored due to the restrictions imposed in \cref{def:qam}. This value, $2^n$, serves as a fundamental limit on the storage capacity of a \gls{qam} based on $n$ qubits.

More in general, a $N$-dimensional Hilbert space admits $N$ orthogonal states. Hence, we define the storage capacity of a \gls{qam} as 
\begin{equation} \label{eq:quantum_storage_capacity}
    \alpha^Q = \frac{M}{\dim \hilbert} \ ,
\end{equation}
where $M$ identifies the number of stored patterns. Similarly to the classical case, analyzing the storage capacity of a \gls{qam} amounts to determining the maximum number of patterns that can be faithfully stored.

Before going ahead, let us remark that, as already mentioned at the end of \cref{sec:framework}, we can devise two different kinds of \gls{qam} tasks, leading to different definitions of storage capacities: (i) in the first case, that we identify as quantum output, we ask the computation to give the requested quantum pattern ($\rho_\mu$) as a direct output, which is not measured at any stage. This approach is suitable for applications where the retrieved quantum information needs to be manipulated or processed further using quantum operations; and (ii) in the second scenario, corresponding to a classical output being related to the quantum state that undergoes a measurement process. The outcome of this measurement is the projected state with a probability depending on the overlap with the corresponding memory. Measurements lead to the identification of the label ($\mu$) associated with the basin of attraction as it would also occur in classification tasks. We stress that the distinctive feature in (ii) is the presence of measurement, while the entire retrieval process, including the evolution map, is ruled by quantum dynamics~\cite{cerezo2022challenges}.

We also note that this distinction is independent of the form of the patterns, which can still be generic quantum states. In the following, we will analyze the critical storage capacity of a \gls{qam} for both scenarios.

\subsection{Quantum output} \label{sec:sc_quantum}

We will now focus on retrieved states that are regarded as quantum outputs. In this case, let us first consider the critical storage capacity when restricting to orthogonal patterns. The result is contained in the following:

\begin{theorem}\label{th:sc_limit}
    The critical storage capacity of a quantum associative memory storing a set of orthogonal patterns is $\alpha_c^{\mathrm{Q}, \perp} = 1/2$, and it is saturated by rank-1 patterns, $\rho_\mu = \op{\psi_\mu}$. 
\end{theorem}
\begin{proof}
    For each pattern, we enforce a non-empty decaying subspace. Furthermore, since we want to store the maximum number of patterns with finite basins of attraction, we set $d_\mu=1$ (the smallest non-vanishing decaying subspace with one element). At this point, a condition that allows one to exploit all the stable subspace for storing patterns consists of taking the minimum rank for each pattern, setting $s_{\mu} = 1$. This in turn implies $ \rho_\mu = \op{\psi_\mu}$. Collecting the above results, the number of stored patterns is $M = N^S = N^D $, and the Hilbert space dimension is $N = N^S + N^D$. Thus, by employing  \cref{eq:quantum_storage_capacity}, it is $\alpha_c^{Q, \perp} = (N/2)/N = 1/2$.
\end{proof}

The criticality of this bound can be understood as follows. Attempting to store more than $N/2$ orthogonal patterns would result in insufficient decaying states for each pattern, and the additional patterns would have no associated basin. Therefore, the extra memory becomes unreachable unless the initial state perfectly matches the desired pattern. Conversely, removing a pattern would create a spurious state in the stable subspace, i.e., a state that remains invariant under the dynamics but is not a memory. In this case, more information could be stored without increasing the dimension of the system.

Our formulation establishes a significant advantage in terms of storage capacity compared to classical associative memories. Going back to the network of $n$ qubits of Hilbert space dimension $2^n$, its critical capacity corresponds to $2^{n-1}$ patterns, which exponentially outperforms classical models. Through a general formulation, we derived (in the previous section) a CPTP map displaying this exponential storage capacity and also presenting finite basins of attraction. We notice that this map improves the capacity $2^{n/2}$ of other \glspl{qam} analysis with finite basins of attraction~\cite{lewenstein2021storage}. 

Beyond orthogonal patterns, we also construct a map [see in \cref{sec:nonrothogonal}] that enables storing an arbitrary number of non-orthogonal states in \glspl{dfs}. However, this does not imply an infinite storage capacity. The reason is that, by requiring a non-empty basin of attraction for each non-orthogonal pattern, we obtain a larger Hilbert space dimension only as a result of an increasing decaying part. On the contrary, in the orthogonal case, the addition of a pattern requires a larger dimension for both the stable and the decaying subspace. 

In the most general case (see \cref{sec:nonrothogonal}) we have $M_\perp$ ($M_\nperp$) orthogonal (non-orthogonal) patterns stored in a stable subspace of dimension $N^S$, which we can divide into $N_\perp^S + N_\nperp^S$. Note that $N_\nperp^S$ is independent of the number of non-orthogonal patterns, while $N_\perp^S$ grows at least linearly with $M_\perp$. Instead, the dimension of the decaying subspace $N^D$ grows for each memory (independent of its type), so that the storage capacity is
\begin{equation} \label{eq:qsc_orthnonorthogonal}
    \alpha^Q = \frac{M_\perp + M_\nperp}{N^S + N^D} = \frac{M_\perp + M_\nperp}{N_\perp^S(M_\perp) + N_\nperp^S + N^D(M)}\ .
\end{equation}
Then, taking into account \cref{th:sc_limit} and assuming $N^D \sim \cO(M)$, the critical storage capacity is
\begin{equation} \label{eq:qsc_orthnonorthogonal_approx}
     \alpha_c^Q \sim \frac{M_\perp + M_{\nperp}}{2 M_\perp + M_{\nperp}}\ .
\end{equation}
As an example, let us consider a stable subspace consisting of only one \gls{dfs}, $\cS = \cX$, with  $\dim \cS = N_\nperp^S = N^S$. The storage of $M=M_\nperp$ patterns requires a finite-dimensional decaying subspace that is at least $M$-dimensional. Since only the dimension of the decaying subspace needs to increase, the storage capacity is $M/(M + N^S)$ which converges to a critical storage capacity of $1$ in the limit of $M \to \infty$.

\subsection{Classical output} \label{sec:sc_class}

An \gls{am} can be used to perform classification tasks for which it is necessary to obtain information about the state at the end of the process. Hence, a measurement of the final quantum state must be performed to determine the pattern. As anticipated in \cref{sec:framework}, due to the phenomenology of measurements in quantum mechanics, one has to repeat the process with multiple copies of the initial state. The resulting statistics provide information on the similarity between the different patterns. This, in turn, allows one to identify which memory is most similar to the final state (as in the classical case), and to gain information on how close the input state was compared to other patterns. This situation is relevant when taking into account non-orthogonal quantum patterns [see \cref{sec:nonrothogonal}], which cannot be perfectly discriminated \cite{chefles2000quantum,bergou2010discrimination}. As a result, this occurrence impacts the storage capacity of a \gls{qam} when employed to retrieve classical information. 

We aim to assess the effect of measurement on the storage capacity for the case outlined above. To this end, let us consider the general case in which patterns are given by $M_\perp$ orthogonal quantum states $\{ \rho_\mu \}$, and $M_{\nperp}$ non-orthogonal quantum states $\{\rho_\ell \}$, with $\tr(\rho_\mu \rho_\ell) = 0\ \forall\mu,\ell$. For the first ones, there exists a projective measurement $\hat{P}_\mu^\perp$, such that $\tr \hat{P}_\mu^\perp \rho_\nu = \delta_{\mu\nu}$, perfectly discriminating the orthogonal states. Instead, for the non-orthogonal states, any measurement displays a finite error probability $P_{\mathrm{err}}$ \cite{helstrom1969quantum}. Although several techniques exist that minimize the discrimination error, these depend on the particular states to be distinguished \cite{rudolph2003unambiguous,pozza2015discrimination}. In general, we will assume that for the set of non-orthogonal patterns $\{\rho_\ell \}$, an optimal \gls{povm} can be found, with a given success probability for discriminating non-orthogonal states, $P_{\mathrm{succ}}^{\mathrm{opt}} = 1 - P_{\mathrm{err}}^{\mathrm{opt}} < 1$. Then, following \cref{eq:qsc_orthnonorthogonal_approx}, the storage capacity becomes
\begin{equation}
    \alpha_c^{QC} \sim \frac{M_\perp + P_{\mathrm{succ}}^{\mathrm{opt}} M_{\nperp}}{2 M_\perp + M_{\nperp}} \ .
\end{equation}
In the limit $M_{\nperp} \to \infty$, the success probability vanishes, as it is not possible to discriminate an infinite number of quantum states within a finite-dimensional state space. In addition, notice that  $\alpha_c^{QC} < \alpha_c^{Q}$, and for any sub-optimal measure we will have $\alpha^{QC} < \alpha_c^{QC}$. This result is consistent with the fact that, when dealing with non-orthogonal patterns,  the amount of information stored is smaller than in the orthogonal scenario, as a result of patterns being correlated. This occurrence is highlighted also in the classical scenarios when considering correlated patterns \cite{gardner1988capacity}, and it has been recently introduced in a continuous-variable system, where the memories are not necessarily orthogonal \cite{labay2022memory}.

\begin{example} \textit{Storage capacity with \gls{dfs}}.
    Consider an \gls{am} with patterns being $M$ \gls{gus}, $\ket{\psi_\mu} = U^{\mu-1} \ket{\psi}$, $\mu=1,...,M$, where $U$ is a unitary operator satisfying $U^M = \Id_\cS$ and the patterns are in general not orthogonal (with $M_{\nperp} =M$). Without loss of generality, the decaying subspace is set to contain $M$ orthogonal states $\{ \ket{\omega_\mu} \}_{\mu=1}^M$, each associated with the corresponding pattern $\ket{\psi_\mu}$. In \cref{sec:gus} we show that we can indeed construct the channel that realizes the association between states in the decaying subspace with the corresponding pattern. Therefore, following \cref{eq:qsc_orthnonorthogonal}, the storage capacity for such a map is $\alpha^Q = M / (N^S + M)$, where $N^S$ is the dimension of the stable subspace. When increasing the number of patterns, the dimension of the stable subspace remains constant. Thus, in the limit $M \to \infty$ we get a finite storage capacity, $\alpha_c^Q \to 1$. Instead, when including a measurement process, we need to account for the maximum success probability of discriminating \gls{gus}. This can be written as $P_{\mathrm{succ}}^{\mathrm{opt}} = N^S / M$, by using the so-called square-root measurement \cite{ban1997gus,eldar2001sqrm}. Therefore, the resulting storage capacity vanishes when increasing the number of patterns, i.e. it is $\alpha_c^{QC} = P_{\mathrm{succ}}^{\mathrm{opt}} \alpha^Q = N^S / (N^S + M) \to 0$ as $M \to \infty$.
\end{example}

\section{Symmetries enabling quantum associative memory} \label{sec:symmetries}

This section explores systems that can be used as quantum associative memories, providing both the patterns and a learning rule. In particular, we discuss how memories can be encoded in steady states of a given open system that exhibits symmetry.

First of all, to enforce the presence of quantum memories, we consider systems that admit multiple steady states, or, equivalently, multiple conserved quantities. This occurrence alone does not provide any insight into the size and shape of decaying spaces $\cD_{\mu}$, and thus on basins of attraction. A sufficient condition for the Hilbert space to separate in invariant subspaces, and related decaying ones is the additional appearance of certain types of symmetries, as we illustrate below. Before going ahead, we point out that there are some differences between how symmetries emerge in Lindbaldian open quantum systems and generic CPTP maps. As such, in the following, we will address both scenarios.

Let us start considering an open quantum system evolving in a Markovian fashion, as described by the \gls{gksl} equation \cite{breuer2002theory}
\begin{equation*}
    \dot{\rho} = \mathcal{L}\rho= -i[\hmt,\rho] + \sum_{\ell}F_{\ell}\rho F_{\ell}^{\dagger} - \frac{1}{2}\lbrace F_{\ell}^{\dagger}F_{\ell}, \rho \rbrace \ .
\end{equation*}
Here, the Liouvillian $\mathcal{L}$ is the generator of the quantum map, $\hmt$ is the Hamiltonian of the system, and $F_{\ell}$ are the so-called jump operators. The system admits multiple steady states if the eigenspace of eigenvalue zero of the Liouvillian, say $\mathcal{L}_{ss}$, is multi-dimensional. The maximum dimension of this eigenspace is $N^2$, where $N = \dim(\hilbert)$. Notably, while all steady states are elements of $\mathcal{L}_{ss}$, the contrary is in general not true \cite{albert2014symmetries}. Before proceeding further, it is convenient to recall that the evolution of a generic operator $O$ reads $\dot{O} = \mathcal{L}^{\dagger}(O) = i[\hmt,O] + \sum_{\ell}F_{\ell}^{\dagger}O F_{\ell} - \frac{1}{2}\lbrace F_{\ell}^{\dagger}F_{\ell}, O \rbrace$, where $\mathcal{L}^{\dagger}$ is the adjoint of the Liouvillian with respect to trace norm. 

In this scenario, if the system displays a strong symmetry,  we can make use of the steady states as quantum memories and, additionally, identify the basins of attraction. Indeed, in a strong symmetry, by definition, there exists a number, say $M$, of operators $J_{\mu}$ that commute with both the Hamiltonian $\hmt$ and the jump operators $F_{\ell}$. If the latter condition holds true, then $\lbrace J_{\mu} \rbrace_{\mu=1}^{M}$ is a set of conserved quantities, $\dot{J}_{\mu} = 0$ $\forall \mu$, and, equivalently, the stationary space $\mathcal{L}_{ss}$ is $M$-dimensional. Moreover, a strong symmetry also induces a weak one (where $J_{\mu}$  commute with $\mathcal{L}$) \cite{buca2012symmetries,albert2014symmetries}, thus allowing one to separate the Hilbert space in symmetry sectors, $\hilbert=\bigoplus_{\mu}\hilbert_{\mu}$. Notably, the evolution inside each $\mu$-th space is split from the others. 

If the $M$ conserved quantities are orthogonal projectors onto the subspace $\hilbert_{\mu}$, the stationary space $\mathcal{L}_{ss}$ does not contain coherences, and each of the subspaces $\hilbert_{\mu}$ hosts a stationary state, which plays the role of the memory. With respect to the notation employed in \cref{sec:construct_qam}, we can thus identify $\cS_{\mu}$, the support of the $\mu$-th stationary state, as well as its corresponding decaying subspace $\cD_{\mu}$ as the orthogonal complement of $\cS_{\mu}$ with respect to $\hilbert_\mu$. In the most general situation, the stationary space $\mathcal{L}_{ss}$ contains also steady state coherences \cite{albert2014symmetries}. In this case, the structure of the stationary state is more complex, as it can host \glspl{dfs} and noiseless subsystems, which have been introduced in \cref{sec:construct_qam}. 

It is worth remarking that there are cases where a set of $M$ conserved quantities, $J_{\mu}$, does not correspond to any symmetry \cite{albert2014symmetries}. These are referred to as dynamical symmetries. Here, both the emergence of the latter and the identification of the decaying space have to be carried out on a case-by-case basis.

Let now consider the case of a generic \gls{cptp} map $\Lambda$. At variance with the previous Markovian case, to derive a separation of the Hilbert space in symmetry sectors, we need to restrict to conserved orthogonal projectors. Indeed, on the one hand, this guarantees multiple steady states, and, on the other hand, it ensures the presence of a global symmetry of the map and of a Hilbert space decomposition. For a more detailed discussion on symmetries and CPTP maps see, e.g., \ccite{albert2019asymptotics}. To our purpose, it is sufficient the following:
\begin{proposition}
    Orthogonal projectors $J_{\mu}$ are conserved quantities, $\Lambda^{\dagger}(J_{\mu})  = J_{\mu}$, iff they commute with the Kraus operators $[J_{\mu}, K_{\alpha}] = 0$, $\forall \alpha, \mu$. 
\end{proposition}
\begin{proof}
    To demonstrate the above result, we proceed similarly to the case of generic conserved quantities (see, e.g. Theorem 5(ii) in \ccite{Blume-Kohout2010pra}). For the sake of a lighter notation, we drop the index $\mu$ in the following. If $[J, K_{\alpha}]=0$ then $\Lambda^{\dagger}(J) = \sum_{\alpha}K_{\alpha}^{\dagger}J K_{\alpha} =  J(\sum_{\alpha}K_{\alpha}^{\dagger} K_{\alpha}) = J $. Conversely, let us assume $\Lambda^{\dagger}(J) = J$. It can be shown that the following relation holds,
    $$\sum_{\alpha} [J,K_{\alpha}]^{\dagger}[J,K_{\alpha}] = \Lambda^{\dagger}[J^{\dagger}J]-J^{\dagger}J \ .$$
    The right-hand side of the above expression vanishes, as $J^{\dagger}J=J$, $J$ being a conserved quantity by assumption; the left-hand side is a sum of positive semidefinite quantities, vanishing iff $[J, K_{\alpha}] =0$. 
   
\end{proof}

As a consequence of the above Proposition, the set $J_{\mu}$ forms an algebra of matrices, which in turn induces a block decomposition of the Hilbert space $\hilbert$. Each state $\rho \in \bh$ can be thus decomposed as a block matrix, the evolution inside each block being separate. It is worth stressing that the conservation of the projectors, (i.e. their commutation with the Kraus operators), also implies the invariance of the CPTP map under $J_{\mu}$, $J_{\mu}\Lambda(\rho)J_{\mu}^{\dagger} = \Lambda(J_{\mu}\rho J_{\mu}^{\dagger}) $, in analogy with a weak symmetry for a Markovian system. In the case $J_{\mu}$ being generic conserved quantities, the above results are restricted to the maximal invariant subspace $\cS$. It can indeed be shown that the Kraus operators $K_{\alpha}^{\mathrm{S}}$ commute with the components $P_S J_{\mu} P_S$, $P_S$ identifying the projector on $\cS$. Hence, the block structure identified through the algebra of matrix of $P_S J_{\mu}P_S$ pertains to the subspace $\cS$ only. As a consequence, the identification of the decaying subspace is not straightforward, in analogy with the case of a dynamical symmetry in the Markovian case. 

\section{Extension to metastable patterns} \label{sect:meta}

All the previous discussion is based on the use of memories encoded into steady states. Nonetheless, it has been shown recently the possibility of realizing a \gls{qam} in a metastable regime \cite{labay2022memory,labay2023squeezed}. This transient memory can speed up retrieval because the decay to the metastable patterns occurs on a shorter time scale than the steady state decay. Therefore, tasks that do not require long-term memory can benefit from this approach. In this section we address this more general case that will be also illustrated in \cref{sec:ex_dd_metastable}.

\begin{figure}
    \centering
    \includegraphics[width=\linewidth]{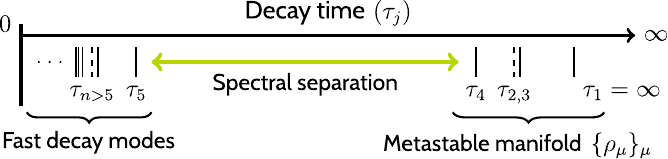}
    \caption{The spectrum of a Liouvillian superoperator with a metastable regime displays a large gap between two consecutive eigenvalues $\lambda_k$. Here we sketch the decay time $\tau_k^{-1} = -\Re[\lambda_k]$ of the eigenvalues. From left to right, we encounter the fast decaying modes with $\tau \ll 1$ separated by a large gap from the four modes defining the metastable manifold which includes the stable state corresponding to $\lambda_1 = 0$.}
    \label{fig:sketch_metastability}
\end{figure}

In general, the evolution of a state undergoing a \gls{gksl} equation can be decomposed as
\begin{equation} \label{eq:mm_general_decomposition}
    \rho(t) = e^{t \lv} \rho_0 = \rho_{\mathrm{ss}}+ \sum_{k \ge 2} c_k e^{t \lambda_k} R_k \ ,
\end{equation}
where $\lambda_k$ are the Liouvillan eigenvalues, sorted for convenience as $0 \le {\Re}[\lambda_k] \le {\Re}[\lambda_{k+1}]$, where $\lambda_1=0$ are the eigenvalues of the stationary state(s) and $R_k$ are the corresponding right eigenmatrices. The coefficients $c_k$ depend on the initial state $\rho(0)$ as $c_k = \tr[L_k^\dagger \rho(t=0)]$ where $L_k$ are the left eigenmatrices. In open quantum systems, metastability is typically observed as a consequence of a gap between two successive eigenvalues of the Liouvillian, say $\lambda_{n}$ and $\lambda_{n+1}$ [see \cref{fig:sketch_metastability}]. After $t > \tau_{n+1}$, all modes with $k > n+1$ give exponentially small contributions to the evolution and can be neglected. Here the system enters the metastable regime, where the evolution of a generic state $\rho$ can be decomposed as
\begin{equation} \label{eq:mm_projector}
    \rho(t) \approx \rho_{\mathrm{ss}} + \sum_{k=2}^n c_k(t) R_k = \cP_{MM}[\rho_0] \ ,
\end{equation}
where $\cP_{MM}$ is a projector superoperator on the first $n$ modes characterizing the metastable manifold. Such a manifold can induce a \gls{dfs} or a noiseless subsystem structure \cite{macieszczak2016towards}, but we will focus here on the so-called \textit{classical metastability}, where each state can be expressed as a convex combination of $n$ disjoint metastable phases $\{ \rho_\mu \}$~\cite{macieszczak2021metastability},
\begin{equation}
    \rho(t) = \sum_{\mu=1}^n p_\mu(t) \rho_\mu  \ ,  
\end{equation}
with $p_\mu(t) \ge 0$ the probability that the state is in the $\mu$-th phase. The metastable phases correspond to physical states that do not evolve in time. The dynamics is fully encoded in the evolution of the probabilities $\{ p_\mu(t) \}$. Moreover, the theory allows us to define a set of projectors $\{ P_\mu \}$ satisfying $\tr P_\mu \rho_\nu = \delta_{\mu\nu}$ and $\sum_\mu P_\mu = \Id$, which define the basin of attraction of each metastable phase. In fact, the value $p_\mu = \tr P_\mu \sigma$ for a general state $\sigma\in \bh$ is preserved until the final decay to the steady state for $t > \tau_\nu$. In analogy to the steady state scenario (see \cref{sec:symmetries}), the metastable phases $\rho_\mu$ act as invariant states for the projective map $\cP_{MM}$ and the projectors $P_\mu$ are the associated conserved quantities. 

Therefore, any quantum system displaying classical metastable dynamics can be understood as an associative memory where (C1) the patterns are the disjoint metastable phases $\{ \rho_\mu \}$, invariant under the map $\cP_{MM}$; (C2) there exists a subspace
\begin{equation}
    \cD_\mu = \{ \ket{\omega} \in \hilbert \mid \ev{P_\mu}{\omega} \ge 1/2 \} \ ,
\end{equation}
containing all the initial conditions evolving under the corresponding pattern with the highest probability; (C3) the subspaces $\cD_\mu$ are approximately orthogonal \footnote{Classical corrections may apply if the metastable phases are not sufficiently orthogonal.}.

The case of classical metastability falls into the category of \gls{qam} with orthogonal patterns, where the number of patterns depends on the dimension of the metastable manifold (the number of slow-decaying Liouvillian modes). Then, the appearance of a gap between the eigenvalues $\lambda_{n}$ and $\lambda_{n+1}$ allows one to store $n$ patterns. The difference between the real parts of the eigenvalues determines the length of the metastable transient. The decay time to the metastable manifold is $\tau_{s} = -[\Re \lambda_{n+1}]^{-1}$, and $\tau_f = -[\Re \lambda_{n}]^{-1}$ determines the end of the metastable regime. Since the dynamics is frozen in between, a measurement can be taken at $\tau_s$ to determine the pattern. The longer we wait to measure, the more likely we are to get a wrong result since jumps between patterns occur at a rate of $1 / \tau_f$. Of course, in the long time limit, all information is lost as the state decays to the final steady state.

\section{Examples} \label{sec:examples}

In this section, we will analyze three examples of \gls{qam} based on the theoretical framework constructed in previous sections. 

\subsection{Quantum random walk} \label{sec:ex_walk}

An interesting approach to \gls{qam} was taken in \ccite{petruccione2014walks}, where the proposed system is a dissipative quantum walk that converges to the predefined patterns in the long time limit. Here, as in the classical \gls{hnn}, the patterns $\{\vx^\mu \}_{\mu=1}^{M_\perp}$ are represented by strings of $n$-bits, $\vx^\mu = (x_1^{\mu},x_2^{\mu},\dots, x_n^{\mu})$ with $x_j^\mu \in \{ 0, 1 \}$, which are encoded in $n$ two-dimensional quantum systems $\hilbert_2$ such that $\vx^\mu \to \ket{x^\mu} = \bigotimes_{i=1}^{n}\ket{x_i^{\mu}}$. Notice that, while the corresponding classical patterns are in general non-orthogonal, this (basis) encoding leads to orthogonal patterns, i.e., quantum states  $\{ \ket{x^\mu} \}$ satisfying $\braket{x^\mu}{x^\nu} = \delta_{\mu\nu}$, so $M = M_\perp$. Thus, the stable subspace is spanned by the patterns, i.e. $\cS = \sspan\{ \ket{x^\mu} \}$, and any other state in $\hilbert_2^{\otimes n}$ belongs to the decaying subspace.

\begin{figure}
    \centering
    \includegraphics[width=\linewidth]{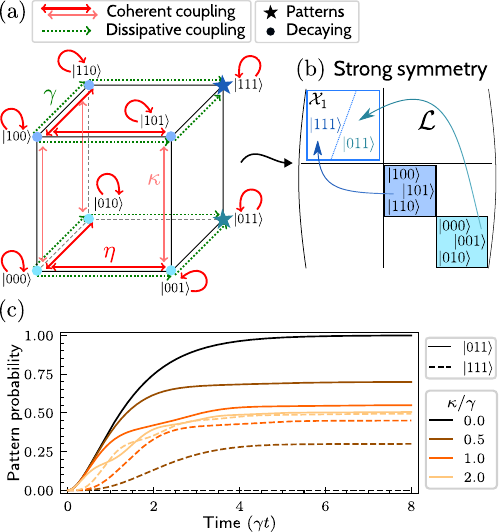}
    \caption{Dissipative quantum walk for \gls{qam}. (a) The nodes of the cube represent the possible states of the walker, and edges are connected between pairs of nodes if their Hamming distance is 1. The coherent dynamics couple states with a strength $\eta$ if they belong to the same basin, and $\kappa$ otherwise. The dissipative part drives the states towards the corresponding fixed points with a rate $\gamma$. (b) Strong symmetry induced in the Liouvillian compatible with \cref{fig:hilber_structure} for $\kappa = 0$. (c) Exact evolution of the initial state $\ket{000}$ for different coupling values between the basins and fixed $\eta/\gamma = 0.1$. The y-axis shows the overlap of the evolving state with the patterns: $\ket{011}$ (solid lines) and $\ket{111}$ (dashed lines). The association to the pattern $\ket{011}$ is perfect only for vanishing $\kappa$ (black lines). For $\kappa > 0$ the symmetry is lost and the probability of retrieving the second pattern increases until $\kappa > \gamma$ where both patterns become equally probable.}
    \label{fig:walk}
\end{figure}

The dynamics can be visualized as a dissipative quantum walk on an $n$-dimensional hypercube [see \cref{fig:walk}(a)], where the nodes are all possible vectors $\vec{\omega} = (\omega_1,\omega_2,\dots,\omega_n) \in \{ 0, 1\}^n$, and edges between two nodes exist if and only if their Hamming distance $d_H(\vec{\omega}, \vec{\omega}') = \sum_{i=1}^n \abs{\omega_i - \omega_i'}/2$ is 1, i.e. they differ only by one bit. Using this definition, we can identify the decaying subspace of the $\mu$-th pattern as the collection of states that are closer to it in terms of Hamming distance than to any other pattern $x^\nu$,  
\begin{equation}
    \cD_\mu = \inbrc{ \ket{\omega} \in \cS^\perp \mid \argmin_{\nu} d_H(\vec{\omega}, \vx^\nu) = \mu }\ .
\end{equation}
Then, the dissipative quantum walk must associate a state $\ket{\omega} \in \cD_\mu$ with $\ket{x^\mu}$. This is engineered by ensuring that (C1) the patterns are invariant states of the Liouvillian superoperator, and (C2) the jump operators, denoted by $L_{\omega\to \omega'}$, couple only states that reduce the Hamming distance towards a pattern,
\begin{equation}
    L_{\omega\to \omega'} = \op{\omega'}{\omega} \text{ if } \min_\mu d_H(\vec{\omega'}, \vx^\mu) < \min_\mu d_H(\vec{\omega}, \vx^\mu) \ .
\end{equation}
A jump between two states happens at a rate $\gamma$, so the walker needs a time $d_H(\vec{\omega},\vx^\mu) / \gamma$ to reach the fixed point. 

The dynamics is completed by a Hamiltonian coupling all nodes at distance 1 except the patterns \cite{petruccione2014walks}. Hence, the Hamiltonian matrix reads
\begin{equation}
    \mel{\omega'}{\hmt}{\omega} = \begin{cases}
        \eta & \ket{\omega},\ket{\omega'} \in \cD_\mu  \\
        \kappa & \ket{\omega} \in \cD_{\mu}\text{ and } \ket{\omega'} \in \cD_{\nu}\ \mu\neq\nu \\
    \end{cases}\ ,
\end{equation}
where nodes are coherently coupled with strength $\eta$ to themselves and to any other node in the same basin, and with strength $\kappa$ to nodes in different basins. The complete Liouvillian is thus $\lv(\rho) = -i [\hmt, \rho] + \sum_{\omega,\omega'} \gamma D[L_{\omega\to \omega'}] \rho$, where $D[O]\rho = O \rho O^{\dagger}-\frac{1}{2} \{ O^{\dagger} O, \rho \}$. The patterns are steady states by construction, and in fact, the stable subspace is a \gls{dfs} where the coherences between patterns are also preserved ($\lv(\op{x^\mu}{x^\nu}) = 0$). 

For instance, in the example of \cref{fig:walk}, there are two patterns $\vx^1 = (0, 1, 1)$ and $\vx^2 = (1, 1, 1)$ with associated decaying subspaces $\cD_1 = \sspan\{\ket{000}, \ket{001}, \ket{010} \}$ and $\cD_2 = \sspan \inbrc{\ket{100}, \ket{101}, \ket{110}}$ respectively. The collection of jump operators produces bit-flips in the first and second qubits to drive the walker towards the patterns, while the coherent Hamiltonian produces oscillations between ($\kappa$) and within ($\eta$) the basins.

Let us first focus on the case $\kappa = 0$. In panel (a) we can see that the cube separates into two disconnected regions containing the states in $\hilbert_1$ (bottom face) and $\hilbert_2$ (top face). Moreover, we can observe the presence of a strong symmetry $P = \sigma_z \otimes \Id_2 \otimes \Id_2$, which means that the overlap with the leftmost qubit is preserved during the evolution. That is, for any initial state $\sigma \in \hilbert$ it holds that $\tr [P \sigma(t = 0)] = \tr [P \sigma(t = \infty)]$, such that if $\sigma \in \hilbert_1$ ($\hilbert_2$) then $\lim_{t\to \infty} e^{\lv t}(\sigma) = \op{x^1}$ ($\op{x^2}$), thus satisfying the associativity condition \eqref{eq:def_assoc_cond}. The associative memory implemented by a quantum random walk with strong symmetry associates a pure state $\ket{\vs}$ with $\vs \in \{ 0,1\}^n$ to the closest pattern $\{ \vx^\mu \}$ in terms of its Hamming distance.  

If the coupling between the basins is $\kappa > 0$, then the strong symmetry is broken, and information can flow between them. As an example, in panel (c) we show the time evolution of the observables $P_{011} = \op{011}$ and $P_{111} = \op{111}$ for the initial state $\ket{000}$. Increasing the value of the coupling $\kappa$ decreases the final retrieval probability since the walker has a non-vanishing probability of going to the other basin. Only the case $\kappa = 0$ leads to perfect retrieval, as expected from the symmetry. Conversely, when the coherent coupling is larger than the dissipation rate ($\kappa /\gamma > 1$), the walker ends up in an equal mixture of both patterns. Therefore, this system is a valid associative memory only in the regime $\kappa/\gamma \sim 0$.

Furthermore, we can calculate the maximum storage capacity of this system. By construction, the model can only store classical-like patterns that are orthogonal. Then, as shown in \cref{sec:orthogonal}, the maximum storage capacity is $1 / 2$. This limit corresponds to storing half of the bit-strings ($2^{n-1}$) as patterns, while the other half belongs to the decaying subspace. For example we can choose as patterns all the states satisfying $\ket{x^\mu} = [\otimes_{i=2}^n \ket{x_i^\mu}] \ket{0}$ (even states), and we can associate them with the odd state $\ket{\omega^\mu} = X_1 \ket{x^\mu} = [\otimes_{i=2}^n \ket{x_i^\mu}] \ket{1}$. The conserved quantity of such a system is the collective spin of the last $(n-1)$ qubits $S = \sum_{i=2}^n \sigma_z^{i}$. 

This example proves the validity of our formulation to explain previous models that use symmetry to realize a \gls{qam}. This particular model allows the perfect association of classical bit strings and may therefore be useful for tasks involving classical data. This system has been realized experimentally using photonic chips \cite{china2019walks}, but other platforms such as ensembles of qubits may also be useful. In fact, it is possible to understand this model as an $n$-spin ensemble, where for the particular example in \cref{fig:walk}(a) the Liouvillian reads $\lv(\rho) = \sum_j \insqr{-i \frac{\omega}{2}\sigma_j^z + \kappa' \sigma_j^x, \rho } + \gamma D[\sigma_1^-]\rho + \gamma D[\sigma_2^-]\rho$ where $\sigma_j^- = \op{1}{0}$ is the decay operator on the $j$-th qubit and $\sigma_j^{x,z}$ are the Pauli matrices. This produces the same dynamics as the dissipative quantum walk and has the patterns as steady states when $\kappa' = 0$ (in fact, it is the example of the amplitude decaying channel on the first two qubits).

The dissipative quantum random walk for \gls{qam} can be applied as a quantum error correction protocol for bit-flip errors. For example, let us consider the three-qubit code, where one has the logical states $\ket{0_L} = \ket{000}$ and $\ket{1_L} = \ket{111}$ \cite{nielsen2010quantum}. Here, single-bit flip errors can be typically corrected by a majority voting mechanism. In the language of associative memory, this means that states that differ by a single bit from one of the two logical states are associated with that bit. Hence, the dissipative quantum walk can be considered as an autonomous quantum error correction for bit-flip errors. Indeed, symmetry protects the system from bit flips by associating errors within the error space with the corrected state. Moreover, since the stable subspace of the system is a \gls{dfs}, quantum coherences are preserved and it allows the storage of a qubit.

\subsection{Driven-dissipative resonator} \label{sec:ex_dd_oscillator}

Driven-dissipative resonators are well-studied systems displaying rich dynamical phenomena like metastability and dissipative phase transitions \cite{minganti2023dissipative} and have been recently employed for quantum memories \cite{mirrahimi2014universal}. In particular, these oscillators have been shown useful for \gls{qam} in a metastable regime leading to improved storage capacities \cite{labay2022memory} and allowing the storage of genuine quantum states \cite{labay2023squeezed}. In this section, we will review the most important characteristics of the metastable \gls{qam} and show the possibility of having permanent memories. Thus highlighting the versatility of these systems for storing different types of quantum patterns.

The oscillator is described by a \gls{gksl} equation
\begin{equation} \label{eq:pp_me_nm}
    \pdv{\rho}{t} = -i [\hmt_n, \rho] + \gamma_1 \disp{\opa}\rho + \gamma_n \disp{\opa^n} \rho\ ,
\end{equation}
where we have standard terms for linear (single-photon) and nonlinear (multiphoton) damping \cite{mundhada2017four,gevorkyan1999coherent} with rates $\gamma_1$ and $\gamma_n$ respectively. The Hamiltonian, which contains a $n$-order squeezing drive \cite{braunstein1987generalized,lang2021multi,minganti2023dissipative}, in the rotation frame and after the parametric approximation is
\begin{equation} \label{eq:pp_ham_nm}
    \hmt_n = \Delta \opad \opa + i \eta \left[\opa^n e^{i \theta_0 n} - (\opad)^n e^{-i \theta_0 n} \right]\ .
\end{equation}
Here, $\Delta = \omega_0 - \omega_s$ is the detuning between the natural oscillator frequency and that of the squeezing force, $\eta$, and $\theta$ the magnitude and phase of the driving, respectively. We observe that the model possesses $\mathds{Z}_n$ symmetry, that is, the transformation $\opa \to \opa \exp(i2\pi/n)$ leaves the master equation invariant \cite{minganti2023dissipative}. The interplay between driving and dissipation, together with the rotation symmetry of the system, leads to the generation of $n$ symmetrically distributed coherent states $\{ \ket{\alpha_j } \}_{j=1}^n$ in the steady state, where $\alpha_j = r \exp(i \theta_j)$ with $r^n = 2\eta / \gamma$ and $\theta_j = 2 \pi j /n + \theta_0$. 

In the following parts, we will study the use of the driven-dissipative oscillator for \gls{qam} in two regimes: the weak symmetry one, where the coherent states are metastable, and the strong symmetry one, where the coherent states form $n$-cat states that are steady states of the dynamics. This allows us to illustrate metastable and stable patterns encoding for \gls{qam}, in terms of coherent states and cat states,  respectively.

\begin{figure*}
    \centering
    \includegraphics[width=\linewidth]{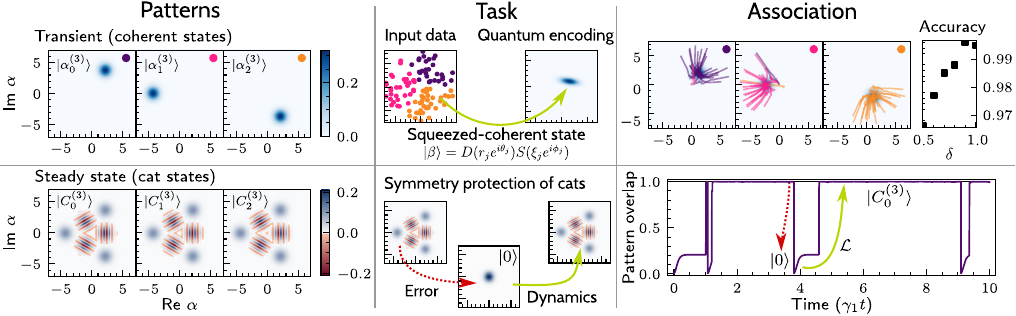}
    \caption{Driven-dissipative resonator for QAM. Parameters: $n=m=3$, $\gamma_3 = 0.2$, $\eta=1.56$, $\Delta=0.4$ and $\theta_0=0$. \textbf{Upper}: Weak symmetry regime ($\gamma_1=1$). (Left) Patterns are defined by the metastable phases: $n$ symmetrically distributed coherent states. (center) Initial data points are encoded in quantum states such as squeezed states, each of which should be associated with the nearest pattern identified by the color (violet $\to \alpha_0$, magenta $\to \alpha_1$ and orange $\to \alpha_2$ ). The basins of attraction [\cref{eq:def_basin_lobes}] that identify the correct pattern are defined with $\delta = 0.5$. (right) The evolution of the state drives the system towards one of the patterns as soon as the system enters the metastable transient. We see the result of the classification of 100 states, where each line is a single trajectory for each initial state in the left part. The upper-right colored dots represent the lobe that the state should be associated with while the trajectory color indicates the measured lobe after association. An accuracy of 0.96 is obtained for $\delta=0.5$ but increasing the spacing between basins leads to higher accuracy as expected. \textbf{Lower}: Strong symmetry regime ($\gamma_1=0$). (Left) Patterns are defined by the steady states of the system: $n$-cat states. (center) The task of the \gls{qam} is to correct amplitude damping errors which bring the system to the empty state $\ket{0}$. (right) Strong symmetry preserves the parity of the original states, so the states converge to the cat state of the corresponding parity.}
    \label{fig:lobes_encoding}
\end{figure*}

\subsubsection{Metastable encoding} \label{sec:ex_dd_metastable}
    In the presence of linear dissipation ($\gamma_1 > 0$), the oscillator has a metastable phase that separates the slowest decaying $n$ Liouvillian modes (including the steady state) from the rest. These $n$ modes define $n$ metastable phases $\{ \rho_\mu \}_{\mu=1}^n$ corresponding to the $n$ lobes forming the steady state, so $\rho_\mu = \op{\alpha_\mu}$. In the upper left part of \cref{fig:lobes_encoding} we see an example for $n=3$, where three symmetrically distributed coherent states define the quantum patterns.
    
    We can also identify a basin of attraction that divides the phase space into $n$ regions. The projector on each of these basins can be written as
    \begin{equation} \label{eq:ex_dd_projector}
        P_\mu = \frac{1}{\pi} \int_{\theta_j - \pi/n}^{\theta_{j} + \pi/n} d\varphi \int_0^\infty dR\ R \op{Re^{i\varphi}} \ .
    \end{equation}
    So all coherent states whose phase is in the interval $[\theta_j - \pi/n, \theta_j + \pi/n]$ will be classified as being in the $j$-th basin. Hence, we may identify the decaying subspaces as
    \begin{equation} \label{eq:def_basin_lobes}
        \cD_\mu = \inbrc{ \ket{\beta} \in \hilbert \mid \tr P_\mu \op{\beta} > \delta \ }\ ,
    \end{equation}
    where $\delta \ge 1/2$. However, since coherent states form an overcomplete basis, the decaying subspaces are disjoint but only approximately orthogonal, as states lying close to the boundary between two basins have a similar probability of converging to both lobes. To prevent this, we may choose a larger value of $\delta$ that would correspond to less distorted patterns.

    In \cref{fig:lobes_encoding}, an oscillator with $n = m = 3$ is used to restore coherent states. The initial states can be obtained from classical data encoded in squeezed coherent states, or directly from quantum inputs. For example, discrete modulated continuous-variable quantum key distribution protocols use symmetric distributed coherent states to encode the keys \cite{Xuan2009cvqkd,norbert2019cvqkd}. The states typically suffer from noise during transmission, which distorts the states. The task is to use \gls{qam} to recover the original states, characterized by the colors of the points in the middle panel. For example, the initial points are randomly chosen coherent states satisfying that the overlap with a basin is at least $\delta = 0.7$ [see \cref{eq:def_basin_lobes}].

    For the reconstruction task, we inject 100 coherent states and compute a single trajectory for each of them. At the start of the metastable transient an unambiguous measurement is performed as explained in \cite{labay2022memory}. The state is associated with a lobe if the measurement triggers the corresponding projector $\Pi_\mu = \op{\alpha_\mu}$ \footnote{The \gls{povm} is completed by the unambiguous operator $\Pi_? = \Id - \sum_\mu \Pi_\mu$ which triggers with a small probability if the lobes are approximately orthogonal, in such case the state is unclassified.}. The trajectory for each lobe is shown in the left panel of \cref{fig:lobes_encoding}. States correctly classified match the color of the trajectory with the plot. In this example, all states are classified, and 96\% of them are reconstructed correctly. The errors are produced by random jumps between lobes and between a lobe and the steady state. The longer we wait to measure inside the metastable regime, the more likely is that a jump will occur. Increasing the value of $\delta \ge 0.8$ leads to a perfect retrieval accuracy, as the basins are more orthogonal, although already with $\delta = 1/2$ we can achieve accuracies above 95\%.

\subsubsection{Cat state encoding}

    When the linear dissipation is turned off ($\gamma_1 = 0$), the system displays a strong $\mathds{Z}_n$ symmetry which divides the Hilbert space into $n$ symmetry sectors. In each sector there exists a steady state  $\rho_\mu$ ($\mu=1,...,n$) and we can identify $n$ conserved quantities $\{ P_\mu \}$ such that $\tr P_\nu \rho_\mu = \delta_{\mu\nu}$. These states for a driven-dissipative nonlinear oscillator correspond to multimode $n$-cat states 
    \begin{equation} \label{eq:multimode_catstates}
        \ket{C_\mu^{(n)}} = \frac{1}{\sqrt{n}} \sum_{k=0}^{n-1} e^{i 2\pi \mu k / n} \ket{\alpha_k}\quad j=0,\dots,n-1 \ ,
    \end{equation}
    where the coherent states $\ket{\alpha_k}$ are the same as in the previous case. For $n = 2$, this reduces to the even and odd cat states $\ket*{C_\pm^{(2)}} = (\ket{+\alpha} \pm \ket{-\alpha})/\sqrt{2}$, and the conserved quantities are the projectors onto the even and odd parity sectors, respectively. In general, the $\mu$-th cat state contains only Fock states, $\ket{n a + \mu}$ where $a \in \mathds{N}$, and the conserved quantities can be expressed as $P_\mu = \sum_a \op{n a + \mu}$. Treating each cat state as a pattern, the strong symmetry enables a \gls{qam} since any state belonging to a symmetry sector will be associated with the corresponding cat state. Even though the Hilbert space is infinite-dimensional, we can also define a stable and decaying subspace. The stable subspace, depending on the system parameters \cite{lieu2020symmetry,minganti2023dissipative}, is spanned either only by the cat states, forming $n$ irreducible subspaces $\{ \cS_\mu \}_{\mu=0}^{n-1}$, or also by their coherences $\op*{C_\mu^{(n)}}{C_\nu^{(n)}}$, forming a three-dimensional \gls{dfs} $\cX$. The latter would allow the storage not only of the multimode cat states but also of any linear superposition of lobes such as a cat state between two lobes, i.e. $\ket{\alpha_0} \pm \ket{\alpha_1}$. Conversely, the decaying subspace associated with each cat state can be constructed through the Gram-Schmidt orthogonalization process such that $\cD_\mu = \{ \ket{\omega} \in \hilbert \mid \ev{P_\mu}{\omega} = 1\text{ and } \braket*{C_\mu^{(n)}}{\omega} = 0 \}$.
    
    This situation resembles the example of the dissipative random walk, where the block structure of the Liovillian generated by a strong symmetry allows for perfect association between decaying and stable states. An example is shown in the lower part of \cref{fig:lobes_encoding} for $n=m=3$. There are three cat states with parity eigenvalue $0$, $2\pi/3$, and $-2\pi/3$ that can be used as patterns. Such multicomponent cat states have been experimentally realized in superconducting platforms \cite{vlastakis2013deterministically} with applications to quantum error correction \cite{peter2016qecmulticat}. 
    
    In the lower part of \cref{fig:lobes_encoding} we perform a particular \gls{qam} task where the patterns are the three-mode cat states [\cref{eq:multimode_catstates} with $n= 3$]. Here, the state may undergo an amplitude damping channel where, with some probability, the state is reset to the ground state $\ket{0}$. Then, the \gls{qam} restores the state. Indeed, since the state is symmetry-protected, the cat state is recovered after a short time, due to its dynamics. An example of such a trajectory is seen in the bottom right panel of \cref{fig:lobes_encoding}, where we plot the overlap with the cat state $\ket*{C_0^{(3)}}$. As soon as the error occurs, the state is associated back to the cat state.

    Note that, measuring the parity of the state is not enough to determine if the evolution has converged to the pattern. Indeed, in this example, parity is always conserved but the state jumps from the vacuum state to the cat state. Hence, in general, the symmetry guarantees that a state will be associated with the steady state with the same symmetry, but we cannot use the respective conserved quantity to determine if the association has happened or not. For that, we must resolve to an operator that can identify if the state has evolved into the pattern or not.

\subsection{Geometrically uniform states} \label{sec:gus}

In this section, we extend the example introduced in \cref{sec:sc_class} to model an \gls{am} that stores many non-orthogonal patterns in a two-dimensional \gls{dfs}. Consider a system of $n + 1$ qubits, where the first one represents the unique stable subspace $\cS = \hilbert_2$ ($\dim \cS = 2$), and the remaining $n$ qubits define the decaying subspace $\cD = \hilbert_2^{\otimes n}$ ($\dim \cD = 2^n$). The patterns stored in $\cS$ correspond to a number $ M_{\nperp} \equiv M$ of \gls{gus}, $\ket{\psi_\ell} = U^{\ell-1} \ket{\psi}$, $\ell = 1,...,M$, where $U$ is a unitary operator, $U = \exp(-i 2 \pi \sigma_y / M)$, so that $U^M = \Id_2$. Then the stable subspace consists of only one \gls{dfs} $\cS = \cX_1$ (with respect to the notation introduced in \cref{sec:nonrothogonal}, $M_\perp = 0$ and $C_{\nperp} = 1$) spanned by the states $\{ \ket{0} \equiv \ket{1_0}, \ket{1} \equiv \ket{1_1} \}$. As explained in \cref{sec:nonrothogonal}, for each pattern there must exist a decaying subspace consisting of all states associated with the particular pattern. Since there is only one \gls{dfs} in the stable subspace, we can decompose the decaying subset as $\cD = \bigoplus_{\ell=1}^{M} \cD_\ell$, where we have omitted the index $\tau=1$, which appears in \cref{eq:structure_decaying_non-orthogonal}, for the sake of easier notation. 

Let us continue with the construction of the decaying subspace $\cD_\ell$ for each pattern. First, a basis for the decaying subspace $\cD$ can be taken as $\lbrace \ket{\omega_x} = \ket{\omega_{x_0}}\ket{\omega_{x_1}} \cdots \ket{\omega_{x_{n-1}}} \rbrace_{x=0}^{2^{n-1}}$, where we denote the elements of such a basis as $x = \sum_{t=0}^{n-1} x_t 2^t$, and the term $\ket{\omega_{x_t}}$ with $x_{t} \in \{0, 1\}$ represents the $t$-th qubit basis and $t=0,...,n-1$. We thus define the $\ell$-th decaying subspace according to
\begin{equation}
    \cD_\ell = \mathrm{span}\{ \ket{x} \in \hilbert_2^{\otimes n} \mid x\mod M = \ell \}\ , 
\end{equation}
so that the decaying states $\ket{\omega_x}$ are associated with the $\ell$-th rotated pattern if the label $x$ modulus $M$ is exactly $\ell$. As an example, if we consider a three-qubit decaying space, $n=3$ and $M=2$ patterns, the two decaying subspaces are $\cD_1 = \mathrm{span}\{ \ket{1}, \ket{3}, \ket{5}, \ket{7}\}$ and $\cD_2 = \mathrm{span}\{ \ket{0}, \ket{2}, \ket{4}, \ket{6}\}$, corresponding to the decomposition $\hilbert_{2}^{\otimes 3} = \hilbert_{2}^{\otimes 2} \bigoplus\hilbert_{2}^{\otimes 2} $. Note, however, that if we consider $M=3$, then we will build decaying subspaces with different dimensions (and this is the case for $M$ odd). Indeed, it is $\cD_1 = \mathrm{span}\{ \ket{1}, \ket{4}, \ket{7}\}$, $\cD_2 =  \mathrm{span}\{ \ket{2}, \ket{5}  \}$, and $  \cD_3 =  \mathrm{span}\{ \ket{0}, \ket{3}, \ket{6}  \}$, so that $\hilbert_{2}^{\otimes 3} = \hilbert_{3} \bigoplus \hilbert_{2}\bigoplus \hilbert_{3} $.

\begin{figure}
    \centering
    \includegraphics{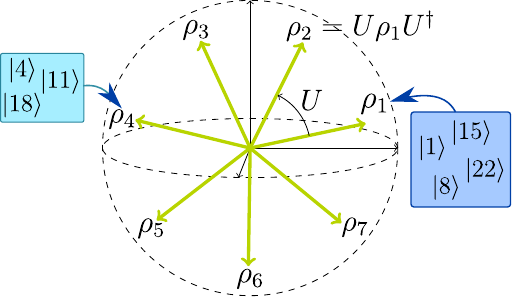}
    \caption{Bloch-sphere representation of the stable subspace characterized by a two-dimensional \gls{dfs} in which seven \gls{gus} defined by the unitary $U$ are stored. For each pattern $\rho_\mu$ there is a decaying subspace formed by states in $\hilbert_2^{\otimes n}$ whose decimal representation module $M$ is equal to $\mu$.}
    \label{fig:gus}
\end{figure} 

The expression for the Kraus operators follows the one given in \cref{eq:kraus_gen_s,eq:kraus_gen_d,eq:kraus_gen_sd}, where the part acting on the stable subspace is proportional to the identity, the part acting on the $\ell$-th decaying subspace reads
\begin{equation}
    K_{\alpha,\ell}^D \equiv K_{\alpha,1,\ell}^D = \sum_{x=0}^{d_\ell -1} c_{\ell, x}^\alpha \op{\omega_{x M + \ell}^\ell} \ ,
\end{equation}
and the part mixing stable and decaying is
\begin{equation}
    K_{\alpha}^{SD} \equiv K_{\alpha,1}^{SD} = \sum_{\ell=1}^M \sum_{j=0}^1 \sum_{x=0}^{d_\ell-1} b_{\ell, j, x}^\alpha \op{j}{\omega_{x M + \ell}^\ell}\ .
\end{equation}
Reminding that we aim at storing \gls{gus} as patterns, i.e. the states
\begin{equation}
    \ket{\psi_\ell} = U^{\ell-1}\ket{\psi} =  \sum_{j=0}^1 \psi_j e^{-i \pi j \ell / M } \ket{j}\ ,
\end{equation}
the associativity condition as formulated in \cref{eq:kraus_gen_assoc_params} reads
\begin{equation}
     \sum_\alpha  (b_{\ell, k, y}^\alpha)^* b_{\ell',j,x}^\alpha  = \delta_{\ell \ell'} \kappa_{xy}^\ell \bra{j}U^{\ell-1} \op{\psi} (U^{\ell'-1})^\dagger\ket{k} \ ,
\end{equation}
which can be expressed as
\begin{equation}
    \sum_\alpha  (b_{\ell, k, y}^\alpha)^* b_{\ell,j,x}^\alpha  = \kappa_{xy}^\ell  \psi_j \psi_k^* e^{-i\pi\ell (j-k)/M} \ .
\end{equation}
By employing the results in \cref{apx:derivation_dfs}, the parameters $\kappa_{xy}^{\ell}$, can be further written as $\kappa_{xy}^{\ell} = \delta_{xy} (\kappa_{x}^{\ell})^2$, where $\kappa_x^\ell$ is given in terms of the decaying parameters $c_{\ell,x}^\alpha$ as in \cref{eq:kraus_assoc_rate_nonorth}. As such, the mixing parameters $b_{\ell, j, x}^{\alpha}$ need to satisfy
\begin{equation}
     \sum_\alpha  (b_{\ell, k, x}^\alpha)^* b_{\ell,j,x}^\alpha  = (\kappa_{x}^\ell)^2  \psi_j \psi_k^* e^{-i\pi\ell (j-k)/M}
     \ .
\end{equation}
Then, a simple solution of the above equation can be derived by associating a single Kraus operator to each collection of non-vanishing mixing terms, $ \{ b_{\ell, j, x}, \forall j \}$, which is uniquely identified by the pair $(\ell, x)$. In this case, the mixing parameters need to satisfy 
\begin{equation}
    b_{\ell,j,x}^\alpha = \psi_j \kappa^\ell_x \exp[i \frac{2\pi}{M}j \ell] \delta_{\alpha-2,xM + \ell} \ .
\end{equation}
First of all, notice that the CPTP condition in \cref{eq:kraus_nonorth_cptp_cond_3} is automatically satisfied with this choice. Secondly, we further took into account that there are two Kraus operators acting only on the stable subspace as $K_\alpha^{\mathrm{S}} = a_\alpha \Id_S$ for $\alpha = 1, 2$ (and $K_\alpha^{\mathrm{S}} = 0$ for $\alpha > 2$) [see \cref{eq:kraus_orth_cptp_cond_2}]. Hence, the final form of the map is
\begin{eqs}
    K_1 &= \insqr{a_1 \Id_S} \oplus \insqr{ \sum_{\ell=0}^{M-1} \sum_{x=0}^{d_\ell-1} \sqrt{1 - \abs{\kappa_x^\ell}^2} \op{\omega_{xM + \ell}}} , \\
    K_2 &= \insqr{a_2 \Id_S} \oplus \mathbb{0} \ , \\
    K_\alpha &= \sum_{j=0}^{1} \psi_j \kappa_x^\ell \exp[i \frac{2\pi j}{M}\ell]\op{j}{\omega_{xM + \ell}} \delta_{\alpha-2, xM+\ell} \ , \quad \nonumber\\&\quad \alpha=3,\dots,2^n+2 \ .
\end{eqs}
The previous map performs perfect association between the decaying states in each basin ($\cD_\ell$) and the corresponding pattern ($\ket{\psi_\ell}$).

Now, we can calculate the storage capacity of the model using \cref{eq:qsc_orthnonorthogonal}. Here, the dimension of the stable subspace is $N^S = 2$, and the dimension of the decaying subspace is $N^D = 2^n$. Then,
\begin{equation} \label{eq:gus_sc}
    \alpha_c^Q = \frac{M}{2 + 2^n} \ ,
\end{equation}
where we note that $M$ is only bounded by $N^D$, as there must be at least a decaying state for each pattern. Thus, in the limit $M \sim 2^n$ we arrive at $\alpha_c^Q \approx 1$.

However, the classical storage capacity is smaller because we have to retrieve the information, that is, we need to discriminate the patterns. The optimal measurement strategy to maximize the success probability of discriminating \gls{gus} is the square-root measurement $P_{\text{succ}} = \abs{\ev{\Phi^{-1/2}}{\psi}}^2$, where $\Phi = \sum_k \op{\psi_k}$ \cite{ban1997gus}. For \gls{gus}, we have $\Phi = (M/N^S) \Id_S$ so $\Phi^{-1/2} = \sqrt{N^S/M}\Id_S$. and $P_{\text{succ}} = N^S / M$. Then, the classical storage capacity is
\begin{equation}
    \alpha_c^{QC} = \frac{N^S}{M} \frac{M}{N^S + N^D} = \frac{N^S}{N^S + N^D} = \frac{2}{2 + 2^n} \ ,
\end{equation}
which is the storage capacity one gets when storing two orthogonal patterns in a two-dimensional stable subspace $\alpha_c^{Q,\perp}$. In other words, due to the measurement, the maximum storage capacity is limited by the storage capacity of orthogonal patterns. But, of course, this storage capacity is much smaller than $\alpha_c^Q$ in \cref{eq:gus_sc}.

\section{Discussion}
\label{sec:disc}

The general framework for quantum associative memory elucidates their key features, underlying principles, and limitations. This unified approach also facilitates meaningful comparisons among diverse recent models. 
Initially, \gls{qam} was framed as a modified version of Grover's search algorithm \cite{ventura2000qam}. This formulation can be regarded as a pattern completion problem rather than an association between the initial state and the target \cite{lloyd2018qhnn}. Patterns are encoded as classical bit strings in $n$ qubits, and the algorithm searches the space of all possible patterns for those that are identical in the first $n-x$ qubits to the input. However, the algorithm is not able to restore imperfect preparation, differing therefore from a \gls{qam}. Moreover, because of its unitary nature, this approach departs from the original Hopfield formulation, as the patterns in question are not fixed points of the dynamics  \cite{ventura2000qam,lloyd2018qhnn}. Indeed,  Grover's algorithm requires an optimal number of iterations depending on the starting point and the number of patterns. A possible solution is to use the modified Grover search method where the target states are fixed points \cite{grover2005fp,tulsi2005new,chuang2014speedupfp}. Interestingly, even though some versions of this algorithm may store an exponential number of classical patterns \cite{ventura2000qam,trugenberger2001probabilistic,TrugenbergerComment,TrugenbergerReply}, it also introduces an exponential number of spurious memories. This is because any state in the Hilbert space has a non-zero probability of retrieval (decreasing with the dimension of the Hilbert space) so any state that is not a pattern is a spurious memory.

More in general, unitary-based \glspl{qam} cannot exhibit fixed points of the dynamics, with the exception of the trivial (identity) case. Hence, these unitary approaches need to estimate beforehand the optimal number of applications of the unitary gate (or the optimal time to evolve a Hamiltonian) to retrieve the desired pattern with the highest probability. For instance, \ccite{trugenberger2001probabilistic} reports that 80 repetitions of the algorithm are needed to retrieve the pattern with probability $1.4\cdot 10^{-4}$. Similarly, the proposals of \cite{cao2017quantum,miller2021quantum} use a repeat-until-success strategy which measures the correct pattern depending on the Hamming distance between the input and all the patterns. As for the capacity of these proposals, it is limited by that of the employed Hebbian learning rule.

Complementary to the quantum circuit-based proposals discussed above, open quantum system approaches (analog) have been recently proposed \cite{rotondo2018open, fiorelli2020signatures, fiorelli2021potts, bodeker2023optimal}. These contributions investigate the conditions under which some quantum many-body spin systems can effectively show associative memory behavior through the use of engineered dissipation. In contrast with unitary proposals, this allows the system to exhibit multiple fixed points, which play the roles of patterns. Notably, they make use of the Hebbian's prescription to embed target spin configurations, i.e. the memories, and therefore the latter can be regarded as classical ones. Within the framework outlined in~\cref{sec:framework}, the quantum map modeling these types of generalizations describes Markovian dynamics, $\Lambda(\cdot) = e^{\mathcal{L} t}(\cdot)$. The evolution governed by the latter is analyzed in these works in specific limiting regimes (e.g., the thermodynamic limit \cite{rotondo2018open, fiorelli2021potts}, or some perturbative regime concerning terms of the dynamical generator \cite{fiorelli2020signatures}) at present. Here, the resulting effective evolution enables the system to operate as an \gls{am}. Within our framework, this can be understood as follows: condition C1, which requires multiple fixed points for the map, is satisfied by employing the Hebbian prescription, which sets classical patterns as fixed points of the effective evolution. Further, the dynamical equations emerging from the latter describe time-dependent variables (rather than operators), i.e., loosely speaking, some classical evolution (see \cite{Fiorelli_2023meanfield} for a rigorous treatment). Such a dynamical system displays finite basins of attraction for each pattern, consistently with condition C2. While these \gls{qam} models in many-body systems are limited to classical regimes, our framework provides a foundation for exploring the quantum regime, presenting a promising avenue for future research.

Looking at further proposals in the existing literature, we identify some models of \gls{qam} that can perform effective association, albeit storing patterns only probabilistically \cite{trugenberger2001probabilistic,cao2017quantum} while other models accomplish perfect storing of patterns in terms fixed points, yet lacking the association property \cite{lewenstein2021storage}. Classically, the two features of stability and association go hand in hand, as can be seen, e.g., in the Hopfield neural network. Here, the non-linear dynamics, equipped with the Hebbian learning rule guarantees that patterns are stable states with non-vanishing basins of attraction, at least below the critical storage capacity. To define a functioning \gls{qam}, however, the stability of patterns and the association between similar states must be individually addressed. In \cref{sec:framework}, we have established the conditions under which we can achieve both features via \gls{cptp} maps.

The two main ingredients in the process of association identified in the discussed framework are dissipation and symmetry. Regarding dissipation, engineered losses act as a mechanism to drive the state of the system into a small subset of long-lived states. This aspect is common to quantum machine learning algorithms such as quantum reservoir computing \cite{sannia2024dissipation,kubota2023qrc}, quantum neural networks \cite{schuld2014quest,beer2020training} or variational quantum algorithms \cite{eisert2024variational,sannia2023variational}, but also in algorithms such as state preparation \cite{Mi2024dissipation,cirac2011preparation}, and quantum error correction \cite{gravina2023critical,gertler2021protecting}. This suggests the usefulness of our framework beyond \gls{qam}. Engineered dissipation can be implemented both in quantum circuits using the collision model algorithm \cite{PhysRevLett.126.130403}, as experimentally demonstrated on the IBM platform \cite{PRXQuantum.4.010324}, and in analog devices using the techniques introduced in \ccite{Verstraete2009}.

Moving forward to the role of symmetries, they allow perfect discrimination between states that fall in the same symmetry sector. Previous research has shown that quantum machine learning models with the same symmetry as the data can avoid training problems and alleviate barren plateaus \cite{cerezo2022group,nguyen2024theory,eisert2023symmetry}. \gls{qam} also takes advantage of symmetries in quantum systems to perform the association process, allowing the storage of stable patterns in open quantum systems featuring strong symmetry or metastable patterns in the case of weak symmetry. For example, in \cref{sec:ex_walk} we show the implementation of autonomous quantum error correction protocols \cite{gravina2023critical,mirrahimi2014universal} where the patterns are fixed points of the dynamics and the error space is corrected dynamically thanks to a strong symmetry by associating the erroneous states with the correct logical qubits \cite{Blume-Kohout2010pra}. In another example proposed in \cref{sec:ex_dd_oscillator}, we show that the metastable phase in a driven-dissipative resonator can be used to correct states generated by discrete-modulated continuous-variable quantum key distribution protocols, both of which share the same discrete rotational symmetry. The developed \gls{qam} framework sets therefore the basis for a broader context of applications, such as quantum memories or error correction, and also contributes to establishing the interplay between dissipation and symmetry in quantum machine learning.

We emphasize that the generality of the proposed formulation allows the storage of arbitrary quantum states as patterns, going beyond the classical-like patterns typically employed in previous formulations of \gls{qam} \cite{ventura2000qam,rotondo2018open,cao2017quantum}. Indeed, the proposed general map can target all possible states, as, for instance, the case of cat states discussed in \cref{sec:ex_dd_oscillator}. By exploiting the general properties of quantum channels, we have constructed a map that is not constrained by the limitations of the classical Hebbian rule. In our case, the storage of quantum patterns is permitted by the learning rule introduced in \cref{eq:kraus_assoc_cond}, which ultimately allows us to deal with quantum data \cite{cerezo2022challenges} and to potentially store an exponentially large number of patterns eﬀiciently. In~\cref{sec:sc} we have separately established the bounds for the storage capacity in the case of quantum outputs, and also in the case of classical outputs, where the effect of measurements was considered. Interestingly, allowing for non-orthogonal patterns can enable an increase of the critical storage capacity for quantum outputs, even though in the case of classical patterns (after measurement) one encounters similar limitations in the association (success) probability, as in classical \gls{am} when correlated patterns are allowed. 

In the case of classification of classical input data,  we can rely on techniques common in quantum machine learning, which use a feature map to encode such data into quantum states \cite{schuld2019feature}. For our approach, and in order to achieve perfect \footnote{Here, we refer to perfect classification as the fact that an input belonging entirely to a class will, at the end of the association process, have components only in the quantum state of that class. } association, it is crucial that data belonging to different labels $\mu$ lead to orthogonal decay states $\ket{\omega^\mu}$, which then evolve into label states $\ket{\mu}$ (orthogonal or not) \cite{MarshallCVZ_19_PRA}. Thus, quantum coding must extract the common features of the $\mu$-th set that are not found in any other class and encode them in our system of interest \cite{lloyd2020embedding}. Of course, finding the best feature map is an open problem in the context of quantum machine learning and beyond the scope of this work. Instead, this work reveals the necessary structure that the final channel must have to perform classification tasks.

In conclusion, we have successfully developed a comprehensive framework for quantum associative memory using CPTP maps. This framework offers several advantages over classical models, such as the ability to encode non-orthogonal states and potentially store an exponentially large number of patterns efficiently. In our formulation, both classical and quantum patterns can be stored, and such patterns can be either stable or metastable dynamical attractors. We also analyzed the role of symmetries, through the definition of basins of attraction, which turn out to be the enabling mechanism for \gls{qam}. These findings open up new possibilities for developing more powerful quantum machine learning algorithms, quantum key distribution, and enhanced error correction techniques. Furthermore they lay the ground to design the most suited implementations in different experimental platforms \cite{labay2023squeezed,MarshEtAl_ccqed_21,fiorelli2020signatures,china2019walks}.

\begin{acknowledgments}
    We acknowledge the Spanish State Research Agency, through the Mar\'ia de Maeztu project CEX2021-001164-M, the COQUSY projects PID2022-140506NB-C21 and -C22  funded by MICIU/AEI/10.13039/501100011033, and by ERDF, EU; MINECO through the QUANTUM SPAIN project, and EU through the RTRP - NextGenerationEU within the framework of the Digital Spain 2025 Agenda. 
     ALM is funded by the University of the Balearic Islands through the project BGRH-UIB-2021.
    EF acknowledges funding by the European Union’s Horizon Europe programme through Grant No. 101105267.
    GLG is funded by the Spanish  Ministerio de Educaci\'on y Formaci\'on Profesional/Ministerio de Universidades and co-funded by the University of the Balearic Islands through the Beatriz Galindo program (BG20/00085). 
\end{acknowledgments}

\bibliography{references.bib}

\appendix

\section{Hopfield Neural Network} \label{apx:hopfield}

Hopfield originally proposed a network of all-to-all connected binary neurons as a content-addressable memory or \gls{am} \cite{hopfield1982am}. Concretely, the model consists of $n$ binary neurons, where the state of the $i$-th neuron, $s_i$, can take two possible values, $s_i=+1$ and $s_i=-1$, corresponding to the firing and resting states, respectively. The state of the compound system can be represented in terms of a $n$-bit string $\vs = (s_1,s_2,...,s_n) $. Further, one can associate the following energy function with the system as follows:
\begin{equation} \label{eq:hopfield_energy}
    E = -\half \sum_{i\neq j} J_{ij} s_i s_j  \ ,
\end{equation}
where $J_{ij}$ represent the coupling between the $i$-th and $j$-th neuron. A deterministic dynamics, evolving the state of the system along trajectories that monotonically decrease the energy function, can be defined in terms of the following single flip of the $i$-th neuron:
\begin{equation} \label{eq:hopfield_update}
    s_i \leftarrow \mathrm{sgn}\inpar{\sum_{j \neq i} J_{ij} s_j } \ .
\end{equation}
In the latter, $\mathrm{sgn}$ is the sign function, letting the $i-$th neuron fire whenever the input signal coming from the other neurons, $\sum_{j \neq i} J_{ij} s_j$,  assumes a positive value. Notice that the input signal plays the role of a local field, $h_i \equiv \sum_{j \neq i} J_{ij} s_j$, acting on the $i$-th neurons. Moreover, the energy function can be expressed in terms of the latter as $E = -\frac{1}{2}\sum_{i}s_i h_i$. Thus, we can see that at each time step the neuron $s_i$ gets aligned with the local field $h_i$, and, correspondingly, the energy undergoes a monotonic decrease, eventually reaching a local minimum, i.e. a stable configuration.

The key point is that stable states can be encoded by means of the coupling parameters $J_{ij}$, in the form of $M$ vectors $\vec{\xi}^\mu$ representing the patterns. Several rules can be used to encode such memories, the Hebbian learning rule being the most prominent example \cite{hebb2005organization}. Here, the couplings are chosen such that
\begin{equation}
    J_{ij} = \frac{1}{n} \sum_{\mu=1}^M \xi_i^\mu \xi_j^\mu\ ,
\end{equation}
and the patterns are further treated as independent and identically distributed random variables. Their distribution is a bimodal one, with $\mathrm{P}[\xi_i^{\mu}=\pm 1] = \frac{1}{2}$. As a result, in the large-$n$ limit patterns are unbiased, $\lim_{n \rightarrow +\infty}\sum_{i}\xi^{\mu}_i/n = 0 $, and uncorrelated $\lim_{n\to\infty} \vec{\xi}^\mu\cdot \vec{\xi}^\nu/n = \delta^{\mu\nu}$.
Under the above conditions, one can show that the patterns $\vec{\xi}^{\mu}$ are stable fixed points of the dynamics \footnote{In fact, also the antipatterns $-\vec{\xi}^\mu$ are stable local minima, and there appear $2M$ configurations that minimize the energy. This is a straightforward consequence of the $\mathds{Z}_2$ symmetry characterizing the energy function.} if $\sqrt{M/n} \ll 1$. In other words, the system behaves as an associative memory. Indeed, patterns are not only fixed points of the dynamics, but they are also stable ones. This means that for each pattern there exists a basin of attraction, i.e. a finite region of the phase space whose points are asymptotically evolved into the pattern itself. 

Hence,  an arbitrary initial state $\vs$ that contains some errors with respect to the patterns is evolved via the dynamics \eqref{eq:hopfield_update} into the most similar one, thus permitting the retrieval of the correct information. 
As already commented in the main text, the maximum number of patterns that can be stored by this kind of associative memory, i.e. its storage capacity, reads $M / n = 0.138$ \cite{amit1985capacity,amit1985spinglass}.

It is worth mentioning that the \gls{hnn} can be further generalized to include some noise, in the form of an effective temperature. The description of the system in this scenario goes beyond the scope of this Appendix, and we refer the reader to some literature on the topic \cite{amit1985capacity,amit1985spinglass}.

\section{Completely Positive Trace-Preserving conditions} \label{apx:cptp_kraus}

A quantum channel is completely positive and trace-preserving if the Kraus operators satisfy the completeness relation \eqref{eq:general_cptp_cond}. Let us enforce this condition for the Kraus operators that are given in the block structure of \cref{eq:general_form_kraus}. Here, the \gls{cptp} condition can be equivalently expressed in terms of the block elements of each Kraus operator, and it reads \cite{albert2019asymptotics}
\begin{eqs}[eqs:kraus_qam_tp_conds]
    &\sum_\alpha (K_\alpha^{\mathrm{S}})^\dagger K_\alpha^{\mathrm{S}} = \Id_{\mathrm{S}} \ ,\label{eq:kraus_qam_tp_s} \\
    &\sum_\alpha (K_\alpha^{\mathrm{S}})^\dagger K_\alpha^{\mathrm{SD}} = \sum_\alpha (K_\alpha^{\mathrm{SD}})^\dagger K_\alpha^{\mathrm{S}} = \mathbb{0}_{\mathrm{SD}} \ ,\label{eq:kraus_qam_tp_sd} \\
    &\sum_\alpha (K_\alpha^{\mathrm{SD}})^\dagger K_\alpha^{\mathrm{SD}} + (K_\alpha^{\mathrm{D}})^\dagger K_\alpha^{\mathrm{D}} = \Id_{\mathrm{D}}\ . \label{eq:kraus_qam_tp_d}
\end{eqs}

In \cref{sec:orthogonal} we introduced the form of the Kraus blocks for associative memory with orthogonal patterns. Using the expression in \cref{eq:kraus_orthogonal_s,eq:kraus_orthogonal_sd,eq:kraus_orthogonal_d} and substituting in \cref{eqs:kraus_qam_tp_conds} we find
\begin{eqs}[eqs:kraus_qam_tp_orth_conds]
    &\sum_\mu\sum_j \insqr{\sum_\alpha\abs{a_{\mu,j}^\alpha}^2}\op{\mu_j} = \Id_{\mathrm{S}} \label{eq:kraus_qam_tp_orth_s} \\
    &\sum_\mu \sum_j \sum_x \insqr{\sum_\alpha (a_{\mu,j}^\alpha)^* b_{\mu,j,x}^\alpha }\op{\mu_j}{\omega_x^\mu} = \mathbb{0}_{\mathrm{SD}} \label{eq:kraus_qam_tp_orth_sd} \\
    &\sum_\mu\sum_{x,y=1}^{d_\mu} \inbrc{\sum_\alpha \insqr{\sum_j (b_{\mu,j,y}^\alpha)^* b_{\mu,j,x}^\alpha} + \delta_{xy}\abs{c_{\mu,x}}^2} \op{\omega_x^\mu}{\omega_y^\mu} \nonumber\\ &\qquad = \Id_{\mathrm{D}} \label{eq:kraus_qam_tp_orth_d}
\end{eqs}
which reduce to the expressions in \cref{eqs:kraus_orth_cptp_cond}. 

When considering non-orthogonal patterns (see \cref{sec:nonrothogonal}) we obtain similar results with the addition of an extra index to account for the \gls{dfs} where the non-orthogonal patterns belong. 

\section{Derivation of Kraus parameters} \label{apx:derivation_kraus}

In this section, we derive the form of the parameters that form the Kraus operators in \cref{sec:construct_qam} for the orthogonal and general formulation of the \gls{qam} map.

\subsection{Orthogonal patterns} \label{apx:derivation_orthogonal}

Combining \cref{eqs:kraus_orth_cptp_cond,eq:kraus_assoc_cond}, we can further express the rate $\kappa_{xy}^{\mu}$ in terms of the coefficients defining the Kraus operators in \cref{eq:kraus_orthogonal_s,eq:kraus_orthogonal_d,eq:kraus_orthogonal_sd}. To do so, as an intermediate step, by exploiting the expressions of $\rho_\mu$ and $K_{\alpha,\mu}^{\mathrm{SD}}$ [see \cref{eq:kraus_orthogonal_s}], one can write \cref{eq:kraus_assoc_cond} in terms of the following system of equations
\begin{equation} \label{eq:kraus_assoc_cond_diag}
    \sum_\alpha (b_{\mu,k,y}^\alpha)^* b_{\mu,j,x}^\alpha = \kappa_{xy}^\mu  u_j^\mu \delta_{jk} \ .
\end{equation}
Thus, by combining the latter with \cref{eq:kraus_orth_cptp_cond_3}, the rates $\kappa_{xy}^\mu$ read
\begin{equation} \label{eq:kraus_assoc_rate_orth}
    \kappa_{xy}^\mu = \delta_{xy}\insqr{1 - \sum_\alpha \abs{c_{\mu,x}^\alpha}^2} \ ,
\end{equation}
where we must have that $0 < \kappa_{xx} < 1$, whence $\sum_\alpha \abs{c_{\mu,x}}^2 < 1$. Therefore, the rate at which the decaying states are associated with the corresponding pattern is the inverse of the rate at which they vanish, consistently with trace preservation. Finally, plugging the last expression for the rates $\kappa_{xy}^{\mu} $ in the constraint \eqref{eq:kraus_assoc_cond_diag}, we obtain the following equation
\begin{equation}
    \sum_\alpha (b_{\mu,k,y}^\alpha)^* b_{\mu,j,x}^\alpha = \insqr{1 - \sum_\alpha \abs{c_{\mu,x}^\alpha}^2}  u_j^\mu \delta_{xy}\delta_{jk} \ .
\end{equation}
A simple solution of the latter exists, assuming that there are as many Kraus operators as the number of combinations of $(j,x)$, characterized only by the non-vanishing parameter $b_{\mu,j,x}^{(j,k)}$. If this is the case, we can write
\begin{equation} \label{eq:kraus_expression_b}
    b_{\mu,j,x}^\alpha = \sqrt{u_j^\mu \insqr{1 - \sum_\alpha \abs{c_{\mu,x}^\alpha}^2}} \delta_{\alpha, j + s_\mu x} \ .
\end{equation}
In the above expression, there is no restriction on the choice of the parameters $c_{\mu,x}^\alpha$ as long as their modulus-squared sum is smaller than one. In terms of the parameters $a_{\mu,j}^\alpha$, we notice that \cref{eq:kraus_orth_cptp_cond_2} imposes a further limit on the Kraus operators displaying $a_{\mu,j}^\alpha \neq 0$. Those Kraus operators with a non-zero $b_{\mu,j,x}^\alpha$ need to have a vanishing $a_{\mu,j}^\alpha$, and vice-versa. Therefore, at least two additional Kraus operators, equipped with a non-vanishing diagonal part of $K_{\alpha,\mu}^{\mathrm{S}}$, need to exist. 

Notice that the solution obtained in \cref{eq:kraus_expression_b} corresponds to the configuration requiring the least amount of Kraus operators. However, such a solution is not unique, as by increasing the number of Kraus operators, one can find other maps $\Lambda$ satisfying the constraints. To conclude, we stress that the parameters $a_{\mu,j}^\alpha$, $b_{\mu,j,x}^\alpha$ and $c_{\mu,x}^\alpha$ can be considered as degrees of freedom that can be tuned to guarantee the associativity condition defined by \cref{eq:def_assoc_cond}. 

\subsection{General formulation} \label{apx:derivation_dfs}

Similar to the previous case, we can find an expression for the mixing parameters for Kraus operators acting on \gls{dfs}. Substituting \cref{eq:kraus_gen_sd} in \cref{eq:kraus_gen_assoc}, we obtain
\begin{equation} \label{eq:kraus_gen_assoc_params}
    \sum_\alpha (b_{\ell',k,y}^{\alpha,\tau})^* b_{\ell,j,x}^{\alpha,\tau} = \delta_{\ell\ell'} \kappa_{xy}^\ell \mel{{\tau}_j}{\rho_{\ell}^{(\tau)}}{{\tau}_k}\ .
\end{equation}
By combining \cref{eq:kraus_gen_sd} and \cref{eq:kraus_gen_d}, the CPTP condition \eqref{eq:kraus_qam_tp_d} (see \cref{apx:cptp_kraus}) reads
\begin{equation} \label{eq:kraus_nonorth_cptp_cond_3}
    \sum_\alpha \insqr{ \sum_{j=1}^{s_\tau} (b_{\ell',j,y}^{\alpha,\tau})^* b_{\ell,j,x}^{\alpha,\tau} } + \delta_{\ell\ell'}\delta_{xy} \abs{c_{\ell,x}^{\alpha,\tau}}^2 = \delta_{\ell\ell'}\delta_{xy}\ .
\end{equation}
Thus, similarly to the derivation of  \cref{eq:kraus_assoc_rate_orth},  an expression for the rate parameters $\kappa_{xy}^\ell$ can be obtained.  Substituting \cref{eq:kraus_gen_assoc_params} into \cref{eq:kraus_nonorth_cptp_cond_3}, we get
\begin{equation} \label{eq:kraus_assoc_rate_nonorth}
    \kappa_{xy}^\ell = \delta_{xy}\insqr{1 - \sum_\alpha \abs{c_{\ell,x}^{\alpha,\tau}}^2}\ ,
\end{equation}
and the expression for the mixing parameters reads
\begin{equation}
    \sum_\alpha (b_{\ell',k,y}^{\alpha,\tau})^* b_{\ell,j,x}^{\alpha,\tau} = \delta_{\ell\ell'}\delta_{xy} \insqr{1 - \sum_\alpha \abs{c_{\ell,x}^{\alpha,\tau}}^2}\mel{{\tau}_j}{\rho_{\ell}^{(\tau)}}{{\tau}_k}\ .
\end{equation}
For the particular case $\rho_{\ell}^{(\tau)} = \op*{\psi_{\ell}^{(\tau)}}$ where $\ket*{\psi_{\ell}^{(\tau)}} = \sum_{j=1}^{s_\tau} [\psi_{\ell}^{(\tau)}]_j \ket{{\tau}_j}$ we may find a solution 
\begin{equation} \label{eq:kraus_expression_b_nonorth}
    b_{\ell,j,x}^{\alpha,\tau} = [\psi_\ell^{(\tau)}]_j\sqrt{1 - \sum_\alpha \abs{c_{\tau,x}^\alpha}^2} \delta_{\alpha - \alpha_0, x m_\nperp^{\mathrm{max}} + \ell} \ .
\end{equation}
with $m_\nperp^{\mathrm{max}} = \max_\tau m_{\nperp}^{(\tau)}$ and $\alpha_0$ is the number of Kraus operators with non-vanishing elements in the stable and decaying part ($\alpha_0 \in [2, (N^S)^2 + (N^D)^2]$). The index of the Kraus operators runs over
$\alpha=1,\dots, m_{max}^\nperp d_{max}$ where $d_{max} = \max_\tau d_{\ell}^{(\tau)}$ since parameters belonging to different $\cX_\tau$ and different $\cS_\mu$ can go in the same Kraus as they are not restricted by the \gls{cptp} conditions.

\section{Relation with genuine incoherent operation} \label{apx:lewenstein}

In this work, we emphasize the necessity that \glspl{qam} feature non-empty decaying subspace, $\cD$. In this way, the states belonging to the latter can be associated with the patterns, i.e. states of the invariant subspace $\cS$, via the dynamical maps of \cref{sec:construct_qam}. In this appendix, we show how such a framework can encompass a model previously proposed \cite{lewenstein2021storage, marconi2022role}, in which a \gls{cptp} map is derived displaying multiple fixed points. Notably, and at variance with our construction, such a map stores an entire basis of the Hilbert space as patterns. Hence, in our formalism, this scenario corresponds to the entire Hilbert space being the stable space, $\cS = \hilbert$, leaving an empty decaying space $\cD = \emptyset$. In fact, as we show in the following, it is possible to derive the above-commented map as a particular case of the \gls{qam} that we obtained in \cref{sec:orthogonal}. 

Let $\{\ket{\mu} \}$ represent a basis of the Hilbert space, with each state $\ket{\mu}$ being invariant under the map $\Lambda$. According to \cref{eq:kraus_orthogonal_s}, the Kraus operators  read
\begin{equation}
    K_\alpha^{\mathrm{S}} = K_\alpha = \sum_\mu a_\mu^\alpha \op{\mu} \ ,
\end{equation}
which, by assumption, defines the Kraus operator acting on the whole Hilbert space $\hilbert$. In the Jamiolkowski-Choi-Sudarshan (JCS) representation \cite{choi1975completely}, the map $\Lambda$ can be expressed in terms of the following linear operator
\begin{equation}
    J_\Lambda = \sum_\alpha \kketbbra{K_\alpha}{K_\alpha} \ ,
\end{equation}
such that $J_{\Lambda} \in \cB(\hilbert \otimes \hilbert)$, and where we have also introduced
\begin{equation}
    \kket{K_\alpha} = \sum_\mu a_\mu^\alpha \kket{\mu\mu} \ ,
\end{equation}
with $\kket{\mu \mu} \in \hilbert \otimes \hilbert$.
By plugging the above in the JCS representation of the map, we obtain
 \begin{align}
    J_\Lambda &= \sum_\alpha \sum_{\mu,\nu} (a_\nu^\alpha)^* a_\mu^\alpha \kketbbra{\mu\mu}{\nu\nu} \nonumber \\
    &= \sum_\mu \kketbbra{\mu\mu}{\mu\mu} + \sum_{\mu\neq\nu} \insqr{\sum_\alpha (a_\nu^\alpha)^* a_\mu^\alpha} \kketbbra{\mu\mu}{\nu\nu} \ ,
\end{align}
having further employed the \gls{cptp} condition of the Kraus operators [\cref{eq:kraus_qam_tp_s}] in the second line. Identifying the term in square brackets with coefficients of the form $1 + \gamma_{\mu\nu}$ we recover the expression of the map proposed \ccite{lewenstein2021storage}. Moreover, being the spectral radius of CPTP maps bounded by $1$, so that $\abs{a_\mu^\alpha} \le 1$, we obtain $\abs{1 + \gamma_{\mu\nu}} \le 1$ if $a_\mu^\alpha \neq a_\nu^\alpha$, this leading to vanishing coherences upon many repetition of the map.
\end{document}